\documentclass[11pt, letterpaper]{article}
\usepackage[utf8]{inputenc}
\usepackage{fullpage,times}

\usepackage[utf8]{inputenc}
\usepackage[english]{babel}
\usepackage{amssymb,amsmath,amsthm,amsfonts}
\usepackage{algorithm}
\usepackage{enumitem}
\usepackage[noend]{algpseudocode}
\usepackage[margin=1in]{geometry}
\usepackage{graphicx}
\newtheorem{theorem}{Theorem}[section]
\newtheorem{claim}{Claim}
\newtheorem{lemma}{Lemma}[section]
\newtheorem*{question*}{Question}

\theoremstyle{definition}
\newtheorem{definition}{Definition}[section]
\newtheorem{remark}{Remark}
\newtheorem*{remark*}{Remark}
\newtheorem*{notation*}{Notation}
\usepackage{mdframed}
\usepackage{subcaption}
\usepackage{graphicx}
\usepackage{bm}
\usepackage{chngcntr}
\counterwithout{table}{section}
\def\multiset#1#2{\ensuremath{\left(\kern-.3em\left(\genfrac{}{}{0pt}{}{#1}{#2}\right)\kern-.3em\right)}}

\usepackage{mathtools}

\title{Lower Bounds for Dynamic Distributed Task Allocation\footnote{A preliminary version of this paper appeared in ICALP 2020.} }
\author{Hsin-Hao Su\\Boston College\and Nicole Wein\footnote{supported by an NSF Graduate Fellowship and NSF Grant CCF-1514339}\\MIT} 

\date{}

\begin{document}

\maketitle

\begin{abstract}
We study the problem of distributed task allocation in multi-agent systems.  
Suppose there is a collection of agents, a collection of tasks, and a \emph{demand vector}, which specifies the number of agents required to perform each task. The goal of the agents is to cooperatively allocate themselves to the tasks to satisfy the demand vector. We study the \emph{dynamic} version of the problem where the demand vector changes over time. Here, the goal is to minimize the \emph{switching cost}, which is the number of agents that change tasks in response to a change in the demand vector. The switching cost is an important metric since changing tasks may incur significant overhead. 

We study a mathematical formalization of the above problem introduced by Su, Su, Dornhaus, and Lynch~\cite{su2017ant}, which can be reformulated as a question of finding a low distortion embedding from symmetric difference to Hamming distance. In this model it is trivial to prove that the switching cost is at least 2. We present the first non-trivial lower bounds for the switching cost, by giving lower bounds of 3 and 4 for different ranges of the parameters.
%The first step of our proofs is to show that the problem of Su et al.~can be reformulated as the problem of finding a low distortion embedding from symmetric difference to Hamming distance. 
%Then we use combinatorial arguments combined with Ramsey theory to derive the structure of low distortion embeddings.

% is to introduce a new combinatorics problem which can be summarized as finding a \emph{low distortion embedding} from symmetric difference to Hamming distance,
%% for permutations of multisets, 
%and show that it is equivalent to the problem of Su et al.
 
\end{abstract}
\clearpage
\thispagestyle{empty}

\pagebreak
\setcounter{page}{1}
\section{Introduction}
%\subsection{Background}
%todo: Say in intro for ITCS (?) innovations are: problem has only been studied in a couple other papers, 1st non-trivial lower bounds for this problem, introduce novel framework that works for proving several of these different results and can maybe hopefully be extended to higher lower bounds?, formal connection to combinatorics of low distortion embedding from symmetric difference to hamming distance (will explain what this means later), should we make an explicit list of these?
Task allocation in multi-agent systems is a fundamental problem in distributed computing.
%, and is especially pertinent to \emph{biological} distributed algorithms such as the division of labor in insect colonies. 
Given a collection of tasks, a collection of task-performing agents, and a \emph{demand vector} which specifies the number of agents required to perform each task, the agents must collectively allocate themselves to the tasks to satisfy the demand vector. This problem has been studied in a wide variety of settings. For example, agents may be identical or have differing abilities, agents may or may not be permitted to communicate with each other, agents may have limited memory or computational power, agents may be faulty, and agents may or may not have full information about the demand vector. See Georgiou and Shvartsman's book~\cite{georgiou2011cooperative} for a survey of the distributed task allocation literature. See also the more recent line of work by Dornhaus, Lynch and others on algorithms for task allocation in ant colonies~\cite{cornejo2014task, su2017ant, dornhaus2018self, radeva2017costs}.

We consider the setting where the demand vector \emph{changes dynamically} over time and agents must redistribute themselves among the tasks accordingly. We aim to minimize the \emph{switching cost}, which is the number of agents that change tasks in response to a change in the demand vector. The switching cost is an important metric since changing tasks may incur significant overhead. Dynamic task allocation has been extensively studied in practical, heuristic, and experimental domains. 
For example, in swarm robotics, there is much experimental work on heuristics for dynamic task allocation (see e.g. \cite{krieger2000ant, csahin2004swarm, mclurkin2005dynamic, mclurkin2004stupid, lerman2006analysis, macarthur2011distributed}). Additionally, in insect biology it has been empirically observed that demands for tasks in ant colonies change over time based on environmental factors such as climate, season, food availability, and predation pressure~\cite{oster1979caste}. Accordingly, there is a large body of biological work on developing hypotheses about how insects collectively perform task allocation in response to a changing environment (see surveys \cite{beshers2001models, robinson1992regulation}). 
%\cite{bonabeau1998fixed}, \cite{franks1994foraging}, 

Despite the rich experimental literature, to the best of our knowledge there are only two works on dynamic distributed task allocation from a theoretical algorithmic perspective. Su, Su, Dornhaus, and Lynch~\cite{su2017ant} present and analyze gossip-based algorithms for dynamic task allocation in ant colonies. Radeva, Dornhaus, Lynch, Nagpal, and Su~\cite{radeva2017costs} analyze dynamic task allocation in ant colonies when the ants behave randomly and have limited information about the demand vector.

%\subsection{Our setting}\label{sec:set}

\subsection{Problem Statement}\label{sec:state}
%Before formally defining the problem statement, we outline the setting that we consider:
%\vspace{-2mm}

We study the formalization of dynamic distributed task allocation introduced by Su, Su, Dornhaus, and Lynch~\cite{su2017ant}.

\subparagraph{Objective:} Our goal is to minimize the \emph{switching cost}, which is the number of agents that change tasks in response to a change in the demand vector. 
%\vspace{-2mm}
\subparagraph{Properties of agents:}
\begin{enumerate}%[itemsep=0mm]
\item the agents have complete information about the changing demand vector
\item the agents are heterogeneous
\item the agents cannot communicate
\item the agents are memoryless
%\item the goal is to minimize the \emph{switching cost}, or the number of agents that change tasks in response to a change in the demand vector.
\end{enumerate}
%\vspace{-2mm}
The first two properties specify \emph{capabilities} of the agents while the third and fourth properties specify \emph{restrictions} on the agents. Although the exclusion of communication and memory may appear overly restrictive, our setting captures well-studied models of both collective insect behavior and swarm robotics, as outlined in Section~\ref{sec:app}. 

From a mathematical perspective, our model captures the \emph{combinatorial} aspects of dynamic distributed task allocation. In particular, as we show in Section~\ref{sec:reform}, the problem can be reformulated as finding a \emph{low distortion embedding} from symmetric difference to Hamming distance. 
%and establishes connections between distributed computation and concepts from classical combinatorics.
%\vspace{-2mm}

%From both a mathematical and practical perspective, it is most desirable to achieve dynamic task allocation algorithms using the \emph{simplest} possible agents. We consider heterogeneous agents with the ability to dynamically detect the number of agents required for each task, but without the powers of \emph{communication} or \emph{memory}.
\subsubsection{Formal statement}
Formally, the problem is defined as follows. There are three positive integer parameters: $n$ is the number of agents, $k$ is the number of tasks, and $D$ is the target \emph{maximum switching cost}, which we define later. The goal is to define a set of $n$ deterministic functions $f_1^{n,k},f_2^{n,k},\dots, f_n^{n,k}$, one for each agent, with the following properties. 
%The function $f_i^{n,k}$ associated with each agent determines which task that agent is assigned to given the current demand vector.
\begin{itemize}
\item {\bf Input:} For each agent $a$, the function $f_a^{n,k}$ takes as input a \emph{demand vector} $\vec{v}=\{v_1,v_2,\dots,v_k\}$ where each $v_i$ is a non-negative integer and $\sum_i v_i=n$. Each $v_i$ is the number of agents required for task $i$, and the total number of agents required for tasks is exactly the total number of agents. 
\item {\bf Output:} For each agent $a$, the function $f_a^{n,k}$ outputs some $i\in [k]$. The output of $f_a^{n,k}(\vec{v})$ is the task that agent $a$ is assigned when the demand vector is $\vec{v}$. 
\item {\bf Demand satisfied:} For all demand vectors $\vec{v}$ and all tasks $i$, we require that the number of agents $a$ for which $f_a^{n,k}(\vec{v})=i$ is exactly $v_i$. That is, the allocation of agents to tasks defined by the set of functions $f_1^{n,k},f_2^{n,k},\dots, f_n^{n,k}$ exactly satisfies the demand vector.
\item {\bf Switching cost satisfied:} The \emph{switching cost} of a pair $(\vec{v},\vec{v'})$ of demand vectors is defined as the number of agents $a$ for which $f_a^{n,k}(\vec{v})\not=f_a^{n,k}(\vec{v'})$; that is, the number of agents that switch tasks if the demand vector changes from $\vec{v}$ to $\vec{v'}$ (or from $\vec{v'}$ to $\vec{v}$). We say that a pair of demand vectors $\vec{v}$, $\vec{v'}$ are \emph{adjacent} if $|\vec{v}-\vec{v'}|_1=2$; that is, if we can get from $\vec{v}$ to $\vec{v'}$ by moving exactly one unit of demand from one task to another. The \emph{maximum switching cost} of a set of functions $f_1^{n,k},f_2^{n,k},\dots, f_n^{n,k}$ is defined as the maximum switching cost over all pairs of adjacent demand vectors; that is, the maximum number of agents that switch tasks in response to the movement of a single unit of demand from one task to another. We require that the maximum switching cost of $f_1^{n,k},f_2^{n,k},\dots, f_n^{n,k}$ is at most $D$.
\end{itemize}

%We are interested in the question of for which values of the parameters $n$, $k$, and $s$, there exists a set of functions $f_1^{n,k},\dots, f_n^{n,k}$ that satisfies the above properties. In particular, we aim to minimize the maximum switching cost:

\begin{mdframed}
\begin{question*}Given $n$ and $k$, what is the minimum possible maximum switching cost $D$ over all sets of functions $f_1^{n,k},\dots, f_n^{n,k}$?
\end{question*}
\end{mdframed}

%There is a series of discrete time steps where at each time step a \emph{demand vector} $\vec{v}=\{d_1,d_2,\dots,d_k\}$  is broadcast to all of the agents, where $d_i$ is the \emph{demand} of task $i$; that is, the number of agents required for task $i$. We require that each demand vector is such that such that $\sum_i d_i=n$ i.e. the number of agents required for tasks is exactly the total number of agents. At each time step, each agent must decide based on $\vec{v}$ which task to perform so that the allocation of agents exactly matches $\vec{v}$. Formally, the output of the problem is $n$ deterministic functions $f_1^{n,k},f_2^{n,k},\dots, f_n^{n,k}$, one for each agent, where each function takes as input a demand vector $\vec{v}$ and outputs a task $i\in [k]$, with following property: for all demand vectors $\vec{v}$ and all tasks $i$, the number of agents $a$ for which $f_a^{n,k}(\vec{v})=i$ is exactly $\vec{v}[i]$; that is, the allocation of agents defined by the set of functions $f_1^{n,k},f_2^{n,k},\dots, f_n^{n,k}$ exactly satisfies the demand vector.
%the task that agent $a$ performs during that time step. The functions must be such that for all tasks $i$, there are exactly $d_i$ agents performing task $i$. 
\subsubsection{Remarks}

\begin{remark}
The problem statement only considers the switching cost of pairs of  \emph{adjacent} demand vectors. We observe that this also implies a bound on the switching cost of non-adjacent vectors: if every pair of adjacent demand vectors has switching cost at most $D$, then every pair of demand vectors with $\ell_1$ distance $d$ has switching cost at most $D(
d/2)$.
% ($d$ is always divisible by 2 since for every demand vector $\sum_i v_i=n$).
\end{remark}

\begin{remark}The problem statement is consistent with the properties of the agents listed above. In particular, the agents have complete information about the changing demand vector because for each agent, the function $f_a^{n,k}$ takes as input the current demand vector. The agents are heterogeneous because each agent $a$ has a separate function $f_a^{n,k}$. The agents have no communication or memory because the \emph{only} input to each function $f_a^{n,k}$ is the current demand vector. 
%for each  because the input to each function $f_a^k$ is Importantly, the input to each function is only a demand vector, which signifies that each agent decides on its task solely based on the current demand vector --- in particular, agents do not make decisions based on communication with other agents or memory of past events. 
%:\vec{v}\rightarrow i\in [k]
\end{remark}

%Given a demand vector that changes over time, the goal is to \emph{minimize} the number of agents that \emph{switch} tasks in response to the changes in the demand vector. Formally, at every time step, the demand vector is changed to an \emph{adjacent} demand vector; that is, exactly one unit of demand is moved from one task to another. (If the demand vector changes more we can simulate it by moving one unit at a time.) The \emph{switching cost} of a pair of adjacent demand vectors $\vec{v_1},\vec{v_2}$ is defined as the number of agents $a$ for which $f_a(\vec{v_1})\not=f_a(\vec{v_2})$. The \emph{switching cost} of a set of functions $f_1^k,f_2^k,\dots, f_n^k$ is defined as the maximum switching cost over all pairs of adjacent demand vectors, that is, the maximum number of agents that switch tasks in response to the movement of a single unit of demand. The following question arises:\\
% i.e. $a$ performs a different task in response to $\vec{v_1}$ and $\vec{v_2}$

%\vspace{-2mm}

\begin{remark} Forbidding communication among agents is crucial in the formulation of the problem, as otherwise the problem would be trivial. In particular, it would always be possible to achieve maximum switching cost 1: when the current demand vector changes to an adjacent demand vector, the agents simply reach consensus about which single agent will move. 
%As such, one technical challenge of this problem is that there are exponentially many ways to move demand over time to reach the present demand vector, however the allocation of agents must be the same regardless of how the present demand vector was reached.
%cannot depend on how demand was moved over time to reach the present demand vector. 
\end{remark}

\subsubsection{Applications}\label{sec:app}

\paragraph{Collective insect behavior}
%There is a recent body of work on applying algorithmic tools to study social insect behavior. For example, a line of recent work of Dornhaus, Lynch and others explores task allocation in ant colonies through an algorithmic lens~\cite{cornejo2014task, su2017ant, dornhaus2018self, radeva2017costs}.

%Su et al.~ first introduced this problem in the context of collective insect behavior.
%It may be reasonable to assume that ants can observe an (approximate) demand vector based on such environmental factors. 
There are a number of hypotheses that attempt to explain the mechanism behind task allocation in ant colonies (see the survey~\cite{beshers2001models}). One such hypothesis is the \emph{response threshold model}, in which ants decide which task to perform based on individual preferences and environmental factors. Specifically, the model postulates that there is an environmental stimulus associated with each task, and each individual ant has an internal threshold for each task, whereby if the stimulus exceeds the threshold, then the ant performs that task. The response threshold model was introduced in the 70s and has been studied extensively since (for comprehensive background on this model see the survey~\cite{beshers2001models} and the introduction of~\cite {duarte2012evolution}).

Our setting captures the essence of the response threshold model since agents are permitted to behave based on individual preferences (property 2: agents are heterogeneous) and environmental factors (property 1: agents have complete information about the demand vector). We study whether models like the response threshold model can achieve low switching costs.

Inspired by collective 
insect behavior, researchers have also studied the response threshold model in the context of swarm robotics~\cite{castello2013task, kim2014response, yang2009swarm}. Our setting also relates more generally to swarm robotics:

%Perform a rigorous analysis of how well ants can mathematically do if their behavior can only depend on individual preferences and environmental factors as is the case for the response threshold model as well as other hypothetical models one could consider.

%Our model captures the response threshold model: stimuli correspond to demand vectors and the the internal thresholds of each ant can be encoded in the set of functions $f_1,f_2,\dots f_n$.

% say isolating the effects of communication/memory

%unfeasible, costly, or unlikely, and 

%don't talk so specifically about applications, consider the general problem and maybe list some applications at the end.
%\vspace{-2mm}
\paragraph{Swarm robotics}
There is a body of work in swarm robotics specifically concerned with property 3 of our setting: eliminating 
the need for 
communication (e.g. \cite{wang2015multi, chen2015cooperation, kernbach2013adaptive, penders2011robot}). In practice, communication among agents may be unfeasible or costly. In particular, it may be unfeasible to build a fast and reliable network infrastructure capable of dealing with delays and failures, especially in a remote location. %Additionally, communication among robots may increase the reaction time of the system to changes in the demand vector.

Regarding property 4 of our setting (the agents are memoryless), it may be desirable for robots in a swarm to not rely on memory. For example, if a robot fails and its memory is lost, we may wish to be able to introduce a new robot into the system to replace it.
%we wish to replace it with a new robot and is subsequently replaced, its replacement a newly arriving robot should ideally be able to determine what task to work on simply by observing the environment.

Concretely, dynamic task allocation in swarm robotics may be applicable to disaster containment~\cite{penders2011robot, zahugi2013oil}, agricultural foraging, mining, drone package delivery, and environmental monitoring~\cite{csahin2004swarm}.

% These scenarios are highly dynamic; that is, the underlying demand vector changes over time. Furthermore, it is reasonable to assume that robots are able to dynamically observe the environment to determine (approximately) how many are needed to contain each distinct geographic area of the disaster.
%Our model could also be relevant to other applications of swarm robotics where instead of the demand vector being observable in the environment, it is broadcast over a network. Such applications may include agricultural foraging, mining, drone package delivery, and environmental monitoring~\cite{csahin2004swarm}.

\subsection{Past Work}\label{sec:past}
Our problem was previously studied only by Su, Su, Dornhaus, and Lynch~\cite{su2017ant}, who presented two upper bounds and a lower bound.
%leaving a huge discrepancy between upper and lower bounds.

The first upper bound is a very simple set of functions $f_1^{n,k},\dots, f_n^{n,k}$ with maximum switching cost $k-1$. Each agent has a unique ID in $[n]$ and the tasks are numbered from 1 to $k$. The functions $f_1^{n,k},\dots, f_n^{n,k}$ are defined so that for all demand vectors, the agents populate the tasks in order from 1 to $k$ in order of increasing agent ID. That is, for each agent $a$, $f_a^{n,k}$ is defined as the task $j$ such that $\sum_{i=0}^{j-1} d_i<\text{ID}(a)$ and $\sum_{i=0}^{j} d_i\geq \text{ID}(a)$. Starting with any demand vector, if one unit of demand is moved from task $i$ to task $j$, the switching cost is at most $|i-j|$ because at most one agent from each task numbered between $i$ and $j$ (including $i$ but not including $j$) shifts to a new task. Thus, the maximum switching cost is $k-1$.
% (recall that $k$ is the number of tasks)

The lower bound of Su et al.~is also very simple. It shows that there does not exist a set of functions $f_1^{n,k},\dots, f_{n}^{n,k}$  with maximum switching cost 1 for $n\geq 2$ and $k\geq 3$. Suppose for contradiction that there exists a set of functions $f_1^{n,k},\dots, f_{n}^{n,k}$ with maximum switching cost 1 for $n=2$ and $k=3$ (the argument can be easily generalized to higher $n$ and $k$). 

Suppose the current demand vector is $[1,1,0]$, that is, one agent is required for each of tasks 1 and 2 while no agent is required for task 3. Suppose agents $a$ and $b$ are assigned to tasks 1 and 2, respectively, which we denote $[a,b,\emptyset]$. Now suppose the demand vector changes from $[1,1,0]$ to the adjacent demand vector $[1,0,1]$. Since the maximum switching cost is 1, only one agent moves, so agent $b$ moves to task 3, so we have $[a,\emptyset, b]$. Now suppose the demand vector changes from $[1,0,1]$ to the adjacent demand vector $[0,1,1]$. Again, since the maximum switching cost is 1, agent $a$ moves from task 1 to task 2 resulting in $[\emptyset,a,b]$. Now suppose the demand vector changes from $[0,1,1]$ to the adjacent demand vector $[1,1,0]$, which was the initial demand vector. Since the maximum switching cost is 1, agent $b$ moves from task 3 to task 1 resulting in $[b,a,\emptyset]$. 

The problem statement requires that the allocation of agents depends \emph{only} on the current demand vector, so the allocation of agents for any given demand vector must be the same regardless of the history of changes to the demand vector. However, we have shown that the allocation of agents for $[1,1,0]$ was initially $[a,b,\emptyset]$ and is now $[b,a,\emptyset]$, a contradiction. Thus, the maximum switching cost is at least 2.
%Again, since the switching cost is 1, $a$ must move to task 2 and $b$ must remain at task 3. Notice that the vectors $[1,1,0]$ and $[0,1,1]$ are adjacent, yet both $a$ and $b$ are assigned to different tasks. Thus, the switching cost is at least 2.

The second upper bound of Su et al.~states that there exists a set of functions $f_1^{n,k},\dots, f_{n}^{n,k}$  with maximum switching cost 2 if $n\leq 6$ and $k=4$. They prove this result by exhaustively listing all 84 demand vectors along with the allocation of agents for each vector. 
% of the allocations of agents to tasks for all %$\multibinom{4}{6}=
%$84$ possible demand vectors. 

\subsection{Our results}
We initiate the study of non-trivial lower bounds for the switching cost.
In particular, with the current results it is completely plausible that the maximum switching cost can always be upper bounded by 2, regardless of the number of tasks and agents. Our results show that this is not true and provide further evidence that the maximum switching cost grows with the number of tasks.

%The second upper bound of Su et al. begs the question: \emph{is it possible to achieve maximum switching cost 2 for $n>6$ or $k>4$?} 

One might expect that the limitations on $n$ and $k$ in the second upper bound of Su et al.~is due to the fact the space of demand vectors grows exponentially with $n$ and $k$ so their method of proof by exhaustive listing becomes unfeasible. However, our first result is that the second upper bound of Su et al.~is actually \emph{tight} with respect to $k$. In particular, 
%they show that one can achieve switching cost 2 for $k=4$ (and $n\leq 6$) and 
we show that achieving maximum switching cost 2 is impossible even for $k=5$ (for \emph{any} $n>2$).
%so it becomes prohibitive to exhaustively list them. 

\begin{theorem}\label{thm:3}
For $n\geq 3$, $k\geq 5$, every set of functions $f_1^{n,k},\dots, f_n^{n,k}$ has maximum switching cost at least 3.
\end{theorem}

We then consider the next natural question: \emph{For what values of $n$ and $k$ is it possible to achieve maximum switching cost 3?} Our second result is that maximum switching cost 3 is not always possible:

\begin{theorem}\label{thm:4}
There exist $n$ and $k$ such that every set of functions $f_1^{n,k},\dots, f_n^{n,k}$ has maximum switching cost at least 4.
\end{theorem}

The value of $k$ for Theorem~\ref{thm:4} is an extremely large constant derived from hypergraph Ramsey numbers. Specifically, there exists a constant $c$ so that Theorem~\ref{thm:4} holds for $n\geq 5$ and $k\geq t_{n-1}(cn)$ where the tower function $t_j(x)$ is defined by $t_1(x) = x$ and $t_{i+1}(x) = 2^{t_i(x)}$.

%We remark that we focus only on small constant values of the switching cost as functions with maximum switching cost 3 already have a highly non-trivial combinatorial structure.

We remark that while our focus on small constant values of the switching cost may appear restrictive, functions with maximum switching cost 3 already have a highly non-trivial combinatorial structure.

%say formal connection between biological distributed algorithms and concepts from classical combinatorics. First step towards obtaining our results is to restate the problem in the language of combinatorics or state an equivalent problem simplifies the analysis with this formulation.
\subsection{Our techniques}\label{sec:tech}

We introduce two novel techniques, each tailored to a different parameter regime. One parameter regime is when $n\ll k$ and the demand for each task is either 0 or 1. This regime seems to be the most natural for the goal of proving the highest possible lower bounds on the switching cost. 

%We develop a proof framework for this regime It is in this setting that we prove Theorem~\ref{thm:4}. 

%On the other hand, although Theorem~\ref{thm:3} considers a lower switching cost than Theorem~\ref{thm:4}, it applies to \emph{all} values of $n\geq 3$ and $k\geq 5$ and its proof requires a different technique. In particular, having an abundance of agents can allow \emph{more} pairs of adjacent demand vectors to have lower switching cost, so it becomes more difficult to find a pair with high switching cost.

\subsubsection{The $n\ll k$ regime} 
We develop a proof framework for the $n\ll k$ regime and use it to prove Theorem~\ref{thm:3} for $n=3$, $k=5$, and more importantly, to prove Theorem~\ref{thm:4}. We begin by supposing for contradiction that there exists a set of functions $f_1^{n,k},\dots, f_{n}^{n,k}$ with switching cost 2 and 3, respectively, and then reason about the structure of these functions. The main challenge in proving Theorem~\ref{thm:4} as compared to Theorem~\ref{thm:3} is that functions with switching cost 3 can have a much more involved combinatorial structure than functions with switching cost 2. In principle, our proof framework could also apply to higher switching costs, but at present it is unclear how exactly to implement it for this setting.

%The $n\ll k$ regime can be burdensome to reason about in the language of the problem staement. To overcome this, we

The first step in our proofs is to reformulate the problem as that of finding a low distortion embedding from symmetric difference to Hamming distance, which we describe in Section~\ref{sec:reform}. This provides a cleaner way to reason about the problem in the $n\ll k$ parameter regime. Our proofs are written in the language of the problem reformulation, but here we will briefly describe our proof framework in the language of the original problem statement.

The simple upper bound of $k-1$ described in Section~\ref{sec:past} can be viewed as each agent having a ``preference" for certain tasks. The main idea of our lower bound is to show that for \emph{any} set of functions $f_1^{n,k},\dots, f_{n}^{n,k}$ with low switching cost, many agents must have a ``preference" for certain tasks. More formally, we introduce the idea of a task being \emph{frozen} to an agent. A task $t$ is frozen to agent $a$ if for every demand vector in a particular large set of demand vectors, agent $a$ is assigned to task $t$. Our framework has three steps:
\begin{itemize}
\item In step 1, we show roughly that in total, many tasks are frozen to some agent. 
\item In step 2, we show roughly that for many agents $a$, only few tasks are frozen to $a$. 
\item In step 3, we use a counting argument to derive a contradiction: we count a particular subset of frozen task/agent pairs in two different ways using steps 1 and 2, respectively.
\end{itemize}

The proof of Theorem~\ref{thm:3} for $n=3$ and $k=5$ serves as a simple illustrative example of our proof framework, while the proof of Theorem~\ref{thm:4} is more involved. In particular, in step 1 of the proof of Theorem~\ref{thm:4}, we derive \emph{multiple} possible structures of frozen task/agent pairs. Then, we use Ramsey theory to show that there exists a collection of tasks that all obey only \emph{one} of the possible structures. This allows us to reason about each of the possible structures  independently in steps 2 and 3.

%Concept of having a task ``frozen" to an agent. It roughly means that for many demand vectors perhaps in a particular restricted set of demand vectors it holds that the agent is assigned to the task. Step 1: Many tasks are frozen to some agent Step 2: for many agents, they can only be frozen to few tasks. Step 3 is to derive a contradiction from aggregating steps 1 and 2 using a counting argument. Count the total number of ``frozen" task/agent pairs. On one hand, this number is large from step 1 and on the other hand this number is small from step 2. For more details see the proof framework section . 

\subsubsection{The remaining parameter regime}\label{sec:remain}
In the remaining parameter regime, we complete the proof of Theorem~\ref{thm:3}. In the previous parameter regime, we only addressed the $n=3$, $k=5$ case, and now we need to consider all larger values of $n$ and $k$. Extending to larger $k$ is trivial (we prove this formally in Section~\ref{sec:3}). However, it is not at all clear how to extend a lower bound to larger values of $n$. In particular, our proof framework from the $n\ll k$ regime immediately breaks down as $n$ grows.

The main challenge of handling large $n$ is that having an abundance of agents can actually allow \emph{more} pairs of adjacent demand vectors to have switching cost 2, so it becomes more difficult to find a pair with switching cost greater than 2. To see this, consider the following example.
% Basically, if a single task $t$ has an abundance of agents, then can always achieve switching cost 2 by moving an agent from the source to $t$ and then from

Consider the subset $S_i$ of demand vectors in which a particular task $i$ has an unconstrained amount of demand and each remaining task has demand at most $n/(k-1)$. We claim that there exists a set of functions $f_1^{n,k},\dots, f_{n}^{n,k}$ so that every pair of adjacent demand vectors from $S_i$ has switching cost 2. Divide the agents into $k-1$ groups of $n/(k-1)$ agents each, and associate each task except $i$ to such a group of agents. We define the functions $f_1^{n,k},\dots, f_{n}^{n,k}$ so that given any demand vector in $S_i$, the set of agents assigned to each task except $i$ is simply a subset of the group of agents associated with that task (say, the subset of such agents with smallest ID). This is a valid assignment since the demand of each task except $i$ is at most the size of the group of agents associated with that task. The remaining agents are assigned to task $i$. Then, given a pair $(\vec{v},\vec{v'})$ of adjacent demand vectors in $S_i$, whose demands differ only for tasks $s$ and $t$, their switching cost is 2 because the only agents assigned to different tasks between $\vec{v}$ and $\vec{v'}$ are: one agent from each of the groups associated with tasks $s$ and $t$, respectively.

Because it is possible for many pairs of adjacent demand vectors to have switching cost 2, finding a pair of adjacent demand vectors with larger switching cost requires reasoning about a very precise set of demand vectors. To do this, we use roughly the following strategy. We identifying a task that serves the role of $i$ in the above example and then successively move demand out of task $i$ until task $i$ is empty and can thus no longer fill this role. At this point, we argue that we have reached a pair of adjacent demand vectors with switching cost more than 2. 

\section{Problem reformulation}\label{sec:reform}

\subsection{Notation} Let $A$ and $B$ be multisets. The intersection of $A$ and $B$ denoted $A\cap B$ is the maximal multiset of elements that appear in both $A$ and $B$. For example, $\{a,a,b,b\}\cap\{a,b,b,c\}=\{a,b,b\}$. The symmetric difference between $A$ and $B$, denoted $A\oplus B$, is the multiset of elements in either $A$ or $B$ but not in their intersection. For example, $\{a,a,b,b\}\oplus\{a,b,b,c\}=\{a,c\}$ since we are left with $a$ after removing $\{a,b,b\}$ from $\{a,a,b,b\}$ and we are left with $c$ after removing $\{a,b,b\}$ from ${a,b,b,c}$. 

A permutation of a multiset $A$ is simply a permutation of the elements of the multiset. For example, one permutation of $\{a,a,b\}$ is $aba$. We treat permutation as strings and perform string operations on them. For strings $X$ and $Y$ (which may be permutations), let $d(X,Y)$ denote the \emph{Hamming distance} between $X$ and $Y$.  For example, $d(aba,bca)=2$.

\subsection{Problem statement}\label{sec:combs}
Given positive integers $n$, $k$, and $D$, the goal is to find a function $\pi_{n,k}$ with the following properties.
\begin{itemize}
\item Let $\mathcal{S}_{n,k}$ be the set of all size $n$ multisets of $[k]$. The function $\pi_{n,k}$ takes as input a set $S\in \mathcal{S}_{n,k}$ and outputs a permutation of $S$.  
\item We say that a pair $S,S'\in \mathcal{S}_{n,k}$ has \emph{distortion} $D'$ with respect to $\pi_{n,k}$ if $|S\oplus S'|=2$ and\\$d(\pi_{n,k}(S),\pi_{n,k}(S'))=D'$. In other words, a pair of multisets has distortion $D'$ if they have the \emph{smallest} possible symmetric distance but \emph{large} Hamming distance (at least $D'$). We say that $\pi_{n,k}$ has \emph{maximum distortion} $D'$ if the maximum distortion over all pairs $S,S'\in \mathcal{S}_{n,k}$ with $|S\oplus S'|=2$ is $D'$. We require that the function $\pi_{n,k}$ has maximum distortion at most $D$.
\end{itemize}
We are interested in the question of for which values of the parameters $n$, $k$, and $D$, there exists $\pi_{n,k}$ that satisfies the above properties. In particular, we aim to minimize the maximum distortion:
\begin{mdframed}
\begin{question*}Given $n$ and $k$, what is the minimum possible maximum distortion over all functions $\pi_{n,k}$?
\end{question*}
\end{mdframed}

In other words, the question is whether there exists a function $\pi_{n,k}$ such that \emph{every} pair $S,S'\in \mathcal{S}_{n,k}$ has distortion at least $D$. Our theorems are lower bounds, so we show that for every function $\pi_{n,k}$ there \emph{exists} a pair $S,S'\in \mathcal{S}_{n,k}$ with distortion at least $D$.

%\vspace{-2mm}
%\paragraph{Problem restatement.} For sets $A$ and $B$, let $A\oplus B$ denote the \emph{symmetric difference} of $A$ and $B$. For strings $a$ and $b$, let $d(a,b)$ be the \emph{Hamming distance} between $a$ and $b$. 

%Given positive integers $n$ and $k$, let $\mathcal{S}_{n,k}$ be the set of all size $n$ multisets of $[k]$. We treat a permutation as a string, and the Hamming distance between two permutations is defined as the Hamming distance between the two corresponding strings. Let $\pi$ be a function that maps every multiset $S\in \mathcal{S}_{n,k}$ to a permutation of $S$. What is the minimum number $D$ such that all pairs $S,S'\in \mathcal{S}_{n,k}$ with $S\oplus S'=2$ satisfy $d(\pi(S),\pi(S'))\leq D$?

%$ $\vspace{-2mm}

\subsection{Equivalence to original problem statement}
We claim that the new problem statement from Section~\ref{sec:combs} is equivalent to the original problem statement from Section~\ref{sec:state}. %That is, we make the following claim:

\begin{claim} 
Given parameters $n$ and $k$ (the same for both problem statements) there exists a function $\pi_{n,k}$ with maximum distortion $D$ if and only if there exists a set of functions $f_1^{n,k},\dots,f_n^{n,k}$ with maximum switching cost $D$.
\end{claim}

We describe the correspondence between the two problem statements:

%We formally prove this claim in the preliminaries/appendix but here, we describe the correspondence. Basically, each $S$ is a demand vector and each $A(S)$ is the allocation of agents to tasks given that demand vector. We describe the correspondence in detail below.
%\vspace{-2mm}
\begin{itemize}%[itemsep=0pt]
\item {\bf Demand vector.} $\mathcal{S}_{n,k}$ is the set of all possible demand vectors since a demand vector is simply a size $n$ multiset of the $k$ tasks. For example, the multiset $S=\{1,1,3\}$ is equivalent to the demand vector $\vec{v}=[2,0,1]$; both notations indicate that task 1 requires two units of demand, task 2 requires no demand, and task 3 requires one unit of demand.
%The fact that $S$ contains two 1s and $\vec{v}$ contains a 2 in the first position both indicate that task 1 requires two units of demand. Similarly, the fact that $S$ contains one 3 and $\vec{v}$ contains a 1 in the third position both indicate that task 3 requires one unit of demand. 
\item {\bf Allocation of agents to tasks.} If $\vec{v}$ is the demand vector representing the multiset $S\in \mathcal{S}_{n,k}$, a permutation $\pi_{n,k}(S)$ is an allocation $f_1^{n,k}(\vec{v}),\dots, f_n^{n,k}(\vec{v})$ of agents to tasks so that $\pi_{n,k}(S)[i]=f_i^{n,k}(\vec{v})$; that is, agent $i$ performs the task that is the $i^{th}$ element in the permutation $\pi_{n,k}(S)$. For example, $\pi_{3,3}(\{1,1,3\})=131$ is equivalent to the following: $f_1^{3,3}([2,0,1])=1$, $f_2^{3,3}([2,0,1])=3$, and $f_3^{3,3}([2,0,1])=1$; both notations indicate that agents 1 and 3 both performs task 1, while agent 2 performs task 2.

\item {\bf Switching cost.} If $\vec{v},\vec{v'}$ are the demand vectors representing the multisets $S,S'\in\mathcal{S}_{n,k}$ respectively, the value $d(\pi_{n,k}(S),\pi_{n,k}(S'))$ is the switching cost because from the previous bullet point, $\pi_{n,k}(S)[i]\not=\pi_{n,k}(S')[i]$ if and only if $f_a^{n,k}(\vec{v})\not=f_a^{n,k}(\vec{v'})$. 

\item %Based on the above correspondences, it is straightforward to see that 
{\bf Adjacent demand vectors.} The set of all pairs $S,S'\in \mathcal{S}_{n,k}$ such that $|S\oplus S'|=2$ is the set of all pairs of adjacent demand vectors. This is because $|S\oplus S'|=2$ means that starting from $S$, one can reach $S'$ by changing exactly one element in $S$ from some $i\in [k]$ to some $j\in [k]$. Equivalently, starting from the demand vector represented by $S$ and moving one unit of demand from task $i$ to task $j$ results in the demand vector represented by $S'$.

%For example, if $S=\{1,1,3\}$ and $S'=\{1,3,3\}$ then $\pi(S)\oplus\pi(S')=2$ and their corresponding demand vectors are $[2,0,1]$ and $[1,0,2]$, which are adjacent.

\item {\bf Maximum switching cost.} If $f_1^{n,k},\dots, f_n^{n,k}$ is the set of functions representing $\pi_{n,k}$, then $\pi_{n,k}$ has maximum distortion $D$ if and only if $f_1^{n,k},\dots, f_n^{n,k}$ has maximum switching cost $D$. This is because $S,S'\in \mathcal{S}_{n,k}$ has distortion $D$ if and only if $|S\oplus S'|=2$ and $d(\pi_{n,k}(S),\pi_{n,k}(S'))=D$ which is equivalent to saying that the demand vectors $\vec{v}$ and $\vec{v'}$ that represent $S$ and $S'$ are adjacent and have switching cost $D$.
\end{itemize}

\subsection{Restatement of results}

We restate Theorems~\ref{thm:3} and \ref{thm:4} in the language of the problem restatement.

\begin{theorem}[Restatement of Theorem~\ref{thm:3}]\label{thm:3r}
Let $n\geq 3$ and $k\geq 5$. Every function $\pi_{n,k}$ has maximum distortion at least 3. 
%there exists a 3-divergent pair $S,S'\in \mathcal{S}_{n,k}$.
%Let $n\geq 3$, let $k\geq 5$, and let $\mathcal{S}_{n,k}$ be the set of all size $n$ multisets of $[k]$. For or every function $\pi$ that maps every multiset $S\in \mathcal{S}_{n,k}$ to a permutation of $S$, there exists a pair $S,S'\in \mathcal{S}_{n,k}$ with $S\oplus S'=2$ and $d(\pi(S),\pi(S'))\geq 3$?
\end{theorem}

\begin{theorem}[Restatement of Theorem~\ref{thm:4}]\label{thm:4r}
There exist $n$ and $k$ so that every function $\pi_{n,k}$ has maximum distortion at least 4.
%Let $n\geq $, let $k\geq $, and let $\mathcal{S}_{n,k}$ be the set of all size $n$ multisets of $[k]$. For every function $\pi$ that maps every multiset $S\in \mathcal{S}_{n,k}$ to a permutation of $S$ there exists a pair $S,S'\in \mathcal{S}_{n,k}$ with $S\oplus S'=2$ and $d(\pi(S),\pi(S'))\geq 4$?
\end{theorem}

\subsection{Example instance}
%\paragraph{Example instance}
%maybe put this example in the intro? and state here that we encourage the reader to look at the example first to get a sense of the problem. Or say here that before outlining the proof methods we do this and say when defining the problem to see section blah for an example instance.
To build intuition about the problem restatement, we provide a concrete example of a small instance of the problem. Suppose $n=3$ and $k=2$. For notational clarity, instead of denoting $[k]=\{0,1\}$ we denote $[k]=\{a,b\}$. Then $\mathcal{S}_{3,2}$ is the set of all size 3 multisets of $\{a,b\}$; that is, $\mathcal{S}_{3,2}=\{\{a,a,a\},\{a,a,b\},\{a,b,b\},\{b,b,b\}\}$. $\pi_{3,2}$ is a function that maps each element of $\mathcal{S}_{3,2}$ to a permutation of itself. For example, $\pi_{3,2}$ could be defined as follows: \[ \begin{matrix*}[l]%
\pi_{3,2}(\{a,a,a\})=aaa,& \pi_{3,2}(\{a,a,b\})=aba& \pi_{3,2}(\{a,b,b\})=bab,& \pi_{3,2}(\{b,b,b\})=bbb.\end{matrix*}\] We are concerned with all pairs $S,S'\in \mathcal{S}_{3,2}$ such that $|S\oplus S'|=2$ (since the maximum distortion of $\pi_{3,2}$ is defined in terms of only these pairs). In this example, the only such pairs are as follows: \[\begin{matrix*}\{a,a,a\}\oplus\{a,a,b\}=2,& \{a,a,b\}\oplus\{a,b,b\}=2,& \{a,b,b\}\oplus\{b,b,b\}=2.\end{matrix*}\] For each such pair, we consider $d(\pi_{3,2}(S),\pi_{3,2}(S'))$: \[\begin{matrix*}d(aaa,aba)=1,&d(aba,bab)=3,&d(bab,bbb)=1.\end{matrix*}\] This particular choice of $\pi_{3,2}$ has maximum distortion 3 (since the largest value in the above row is 3), however we could have chosen $\pi_{3,2}$ with maximum distortion 1 (for example if $\pi_{3,2}(\{a,b,b\})=bba$ instead of $bab$).

\section{The $n\ll k$ regime}\label{sec:n<k}
In this section we will prove Theorem~\ref{thm:3r} for $n=3$, $k=5$, and Theorem~\ref{thm:4r}. The proofs are written in the language of the problem reformulation from Section~\ref{sec:reform}. For these proofs it will suffice to consider only the elements of $\mathcal{S}_{n,k}$ that are \emph{subsets} of $[k]$, rather than multisets. This corresponds to the set of demand vectors where each task has demand either 0 or 1. \emph{For the rest of this section we consider only subsets of $[k]$, rather than multisets.}

We call each element of $[k]$ a \emph{character} (e.g. in the above example instance, $a$ and $b$ are characters).

%This regime is good for proving higher lower bounds on the switching cost.  This regime seems to be the best for proving higher lower bounds on the switching cost. As discussed above, having large demands seem to make it more difficult to prove lower bounds. In this section we will prove Theorems  for specific values of the parameters with $n<k$. In the next section, we will complete the proof of Theorem by using entirely different proof techniques designed to handle large $n$. 
%
%More specifically, in this section it will suffice to consider demand vectors where each task has demand either 0 or 1. In the language of the problem reformulation, this corresponds to considering only the subset of $\mathcal{S}_{n,k}$ that consists of all size $n$ \emph{subsets} of $[k]$ (rather than multisets).

\subsection{Proof framework}\label{sec:frame}

%In general, the set of pairs of multisets with symmetric difference 2 will of course have a much more complicated structure).
%In this section we present a framework that we use to prove our lower bounds.

%We present three proofs in this paper in order from simplest to most involved. The first proof is for the special case of Theorem~\ref{thm:3r} where $n=3$, 
%This warm-up result does not immediately extend to $n\geq 3$ as described in the previous subsection. 
%the second proof is for the remaining case of Theorem~\ref{thm:3r} where $n>3$, and the third proof is for Theorem~\ref{thm:4r}. All three proofs follow the same general three-step framework, which we outline below. 
%The proof of Theorem~\ref{thm:warm} is the simplest possible illustrative example of this framework.
% and we encourage the reader to reference that proof to gain a concrete understanding of the abstract description of the framework in this section.

%We describe the framework in an abstract sense in this section

%It might be easier to understand the framework with this example than with the abstract description. 

%Present it as kind of a framework where it's like step 1: consider sets of size n-1 and all the ways they could overlap 2: consider sets of size n-2 and the corresponding sets of size n-1 and show that in some sense they must be consistent with each other (because if you combine them and they're not consistent then get contradiction). Step 3: Use a counting argument to show that there must exist multiple symbols that should be in the same place.

As described in Section~\ref{sec:tech}, we develop a three-step proof framework for the $n\ll k$ regime. Suppose we are trying to prove that every function $\pi_{n,k}$ has maximum distortion at least $D$ for a particular $n$ and $k$. We begin by supposing for contradiction that there exists $\pi_{n,k}$ with maximum distortion less than $D$. That is, we suppose that every pair $S,S'\in \mathcal{S}_{n,k}$ with $|S\oplus S'|=2$ has $d(\pi_{n,k}(S),\pi_{n,k}(S'))<D$. Under the assumption that such a $\pi_{n,k}$ exists, steps 1 and 2 of the framework show that $\pi_{n,k}$ must obey a particular structure. For the remainder of this section, we drop the subscript of $\pi$ since $n$ and $k$ are fixed.

\begin{notation*} For any set $R\subseteq [k]$, let $\mathcal{U}_R$ be the set of all sets $S\subseteq [k]$ such that $R\subset S$ and $|S|=|R|+1$.
\end{notation*}

\paragraph{Step 1: Structure of size $n-1$ sets.} 
We begin by fixing a size $n-1$ set $R\subseteq [k]$. Now, consider $\mathcal{U}_R$ (defined above). We note that all pairs $S,S'\in \mathcal{U}_R$ are by definition such that $|S\oplus S'|=2$. Because we initially supposed that $\pi$ has maximum distortion less than $D$, we know that for all pairs $S,S'\in\mathcal{U}_R$, we have $d(\pi(S),\pi(S'))<D$. 

%and considering applying $\pi$ to each of the size $n$ supersets of $P$. Becuase this is a set of very similar multisets, we can conclude that their corresponding permutations are also very similar.

Then we prove a structural lemma which roughly says that many characters $r\in R$ have a ``preference" to be in a particular position in the permutations $\pi(S)$ for $S\in \mathcal{U}_R$. We say that $R$ \emph{$i$-freezes} the character $r$ if $\pi(S)[i]=r$ for many $S\in \mathcal{U}_R$. Our structural lemma roughly says that for many characters $r\in R$, there exists an index $i\in [n]$ such that $R$ $i$-freezes $r$. In other words, for many $S\in \mathcal{U}_R$, the $\pi(S)$s agree on the position of many characters in the permutation.  
% The exact statement of the structural lemma and the definition of ``$i$-freeze" are different for each of our proofs.

\paragraph{Step 2: Structure of size $n-2$ sets.} We begin by fixing a size $n-2$ set $Q\subseteq [k]$. Now, consider $\mathcal{U}_Q$. We note that each $R\in\mathcal{U}_Q$ obeys the structural lemma from step 1; that is, for many characters $r\in R$, there exists an index $i\in [n]$ such that $R$ $i$-freezes $r$. 

%consider all size $n-2$ subsets of $[k]$. Let $Q$ be such a multiset. We consider the set $\mathcal{U}_Q$ of all size $n-1$ multisupersets of $Q$ from $[k]$. 
%We note that for every pair $P,P'\in \mathcal{U}_Q$, by definition $|P\cup P'|=n$. Thus, $\pi(P\cup P')$ is well-defined. 

%We build upon the structural lemma from step 1 to prove a structual lemma about
We prove a structural lemma which roughly says that the sets $P\in \mathcal{U}_Q$ are for the most part \emph{consistent} about which characters they freeze to which index of the permutation. 
%That is, for many characters $q\in Q$, $q$ can only be $i$-frozen for a single value of $i$. 
More specifically, for many characters $q\in Q$, for all pairs $P,P'\in \mathcal{U}_Q$, if $R$ $i$-freezes $r$ and $R'$ $j$-freezes $r$, then $i=j$. %This implies roughly that for each index $i\in [n]$, only few characters can be $i$-frozen in aggregate over all $P\in \mathcal{U}_Q$. 

\paragraph{Step 3: Counting argument.}
In step 3, we use a counting argument to derive a contradiction. 
%To do so, we count the same quantity in two different ways. 
For the proof of Theorem~\ref{thm:3r}, a simple argument suffices. The idea is that step 1 shows that many characters are frozen overall while step 2 shows that each character can only be frozen to a single index. Then, the pigeonhole principle implies that more than one character is frozen to a single index, which helps to derive a contradiction.

For the proof of Theorem~\ref{thm:4r}, it no longer suffices to just show that more than one character is frozen to a single index. Instead, we require a more sophisticated counting argument and a careful choice of what quantity to count. We end up counting the number of pairs $(Q,a)$ such that $R\in \mathcal{U}_Q$, where $Q\subset [k]$ is a size $n-2$ set and $a\in [n]\setminus Q$. To reach a contradiction, we count this quantity in two different ways, using steps 1 and 2 respectively.

Having reached a contradiction, we conclude that $\pi$ has maximum distortion at least $D$. 

%We essentially count the total number of characters that are frozen in aggregate over all size $n-1$ sets. The structural lemma from step 1 provides a lower bound on this quantity by saying that \emph{every} size $n-1$ set freezes \emph{many} characters. On the other hand, the structural lemma from step 2 provides an upper bound on this quantity by saying roughly that for every set $Q$ of size $n-2$, only \emph{few} characters can be frozen to each index by some set in $\mathcal{U}_Q$.

%The final contradiction that we reach is that the lower bound implied by step 1 is higher than the upper bound implied by step 2. This proves that the original supposition is false, and we conclude that $\pi$ has maximum distortion at least $D$. 

%Because of the strict structure imposed by steps 1 and 2 counting the total number of characters $p\in [n]$ that are $i$-frozen. Like many things must get frozen overall but also they must be consistent with the structure. Given a multiset of size $n-2$ only a limited number of characters of $[k]$ can be frozen to each index over all multisupersets of size $n-1$ by the structural lemma from step 2. From structual of step 1, many must get frozen overall, and from structual of step 2, few things must get frozen overall. Use a counting argument to compute this a get a contradiction on the original supposition that less than D divergent.

%Apply pigeonhole principle to say that there exists a symbol mapped to two different places or multiple symbols mapped to the same place in a single permutation.

\subsection{Proof of Theorem~\ref{thm:3r} for $n=3$, $k=5$}

%We begin by considering a very special case of the problem as an illustrative example of our proof framework. In this section we prove the simplest possible special case of Theorem~\ref{thm:3r}.

%The purpose of this is to introduce the techniques. We don't actually reference this in the other proof. We follow the framework introduced in the overview section.
%say we break into two sections and the sections taken together complete the proof
%should we note that 3/5 case could be exhaustively searched on a computer (or with simpler pf) but this pf give the intuition behind why it's not possible and introduces framework, and the switching cost 3 case becomes too big to do an exhaustive search.

In this section, we prove Theorem~\ref{thm:3r} for $n=3$, $k=5$, which serves as a simple illustrative example of our proof framework from Section~\ref{sec:frame}.

\begin{theorem}[Special case of Theorem~\ref{thm:3r}]\label{thm:warm}
Every function $\pi_{3,5}$ has maximum distortion at least 3. 
\end{theorem}
\begin{proof}
%We will prove the theorem only for $k=5$. Then, by Lemma , this immediately implies the result for $k\geq 5$. 

%Although there are ways to prove Theorem~\ref{thm:warm} that are perhaps simpler, we use this simple setting of parameters as an illustrative example of our proof framework.

Suppose by way of contradiction that there is a function $\pi_{3,5}$ with maximum distortion at most 2. For the remainder of this section we omit the subscript of $\pi$ since $n=3$, $k=5$ are fixed. For clarity of notation, we let $\{a,b,c,d,e\}$ be the characters in $[k]$ for $k=5$. Thus, we are considering the set of all ${5\choose 3}=10$ size 3 subsets of $\{a,b,c,d,e\}$. (Recall that we are only concerned with subsets, not multisets.)

%is this pi notation ok because we're using it also there exists a function pi and then fixing it, eh it's probably fine.

%The $3$-divergent pair $S,S'\in \mathcal{S}_{3,5}$ that we will find will be such that both $S$ and $S'$ are \emph{subsets} of $[k]$ instead of just multisets of $[k]$. Moreover, for the whole proof it will suffice to consider subsets instead of multisets.\footnote{For context, the special case of considering only subsets instead of multisets corresponds to the special case of the original problem formulation where each task has demand at most 1 at all times.} Because we only consider subsets instead of multisets, this proof is in some ways simpler than the more general description presented in the proof framework.
% move above description to further up since this whole framework is only for sets.

\subsubsection*{Step 1: Structure of size $n-1$ sets} 
%The structural property that we prove for step 1 is stated in Lemma~\ref{lem:fix}. First, we state some definitions.
We begin by fixing a set $\{x,y\}\subseteq \{a,b,c,d,e\}$ of size $n-1=2$. Recall that $\mathcal{U}_{\{x,y\}}$ is the set of all size 3 sets $S$ such that $\{x,y\}\subseteq S\subseteq \{a,b,c,d,e\}$. For example, $\mathcal{U}_{\{a,b\}}=\{\{a,b,c\},\{a,b,d\},\{a,b,e\}\} $. We note that by definition all pairs $S,S'\in \mathcal{U}_{\{x,y\}}$ have $|S\oplus S'|=2$. Thus, to find a pair with distortion 3 and thereby obtain a contradiction, it suffices to find a pair $S,S'\in \mathcal{U}_{\{x,y\}}$ with Hamming distance $d(\pi(S),\pi(S'))=3$. Since $n=3$, this means we are looking for permutations $\pi(S),\pi(S')$ that disagree about the position of \emph{all} elements.

%Consider $\pi(\mathcal{U}_p)$. 
%For all $U\in \mathcal{U}_{\{x,y\}}$, consider $\pi(U)$. 
The following lemma says that $\pi$ places one of $x$ or $y$ at the \emph{same} position for \emph{all} $\pi(S)$ with $S\in \mathcal{U}_{\{x,y\}}$. For ease of notation, we give this phenomenon a name:

\begin{definition}[freeze]\label{defn:freeze} We say that a pair $\{x,y\}\subseteq \{a,b,c,d,e\}$ \emph{$i$-freezes} a character $p\in \{x,y\}$ if for all $S\in\mathcal{U}_{\{x,y\}}$, we have $\pi(S)[i]=p$. We simply say that $\{x,y\}$ \emph{freezes} $p$ if $i$ is unspecified.   Equivalently, we say that a character $p$ is \emph{$i$-frozen} (or just \emph{frozen}) by a pair.
\end{definition}

\begin{lemma}\label{lem:fix} 
For every $\{x,y\}\subseteq \{a,b,c,d,e\}$, there exists $i$ so that $\{x,y\}$ $i$-freezes either $x$ or $y$.

%For every pair $\{x,y\}\subseteq \{a,b,c,d,e\}$, there is some $p\in \{x,y\}$ and some $i\in [n]$ such that for all $U\in \mathcal{U}_{\{x,y\}}$, we have $\pi(U)[i]=p$. 
\end{lemma}

For example, one way that the pair $\{a,b\}$ could satisfy Lemma~\ref{lem:fix} is if the permutations $\pi(\{a,b,c\})$, $\pi(\{a,b,d\})$, and $\pi(\{a,b,e\})$ \emph{all} place the character $a$ in the $0^{th}$ position. In this case, we would say that the pair $\{a,b\}$ 0-freezes $a$.

\begin{proof}[Proof of Lemma~\ref{lem:fix}]
Without loss of generality, consider $\{x,y\}=\{a,b\}$. In this case, $\mathcal{U}_{\{x,y\}}=\mathcal{U}_{\{a,b\}}=\{\{a,b,c\},\{a,b,d\},\{a,b,e\}\} $. Thus, we are trying to show that $\{a,b,c\}$, $\{a,b,d\}$, and $\{a,b,e\}$ \emph{all} agree on the position of either $a$ or $b$.

Suppose without loss of generality that $\pi(\{a,b,c\})=abc$. We first note that $\pi(\{a,b,c\})$ and $\pi(\{a,b,d\})$ must agree on the position of either $a$ or $b$ because otherwise we would have $d(\pi(\{a,b,c\}), \pi(\{a,b,d\}))=3$ which would mean that $\pi(\{a,b,c\})$ and $\pi(\{a,b,d\})$ would have distortion 3, and we would have proved Theorem~\ref{thm:warm}. Without loss of generality, suppose $\pi(\{a,b,c\})$ and $\pi(\{a,b,d\})$ agree on the position of $a$; that is, $\pi(\{a,b,d\})$ is either $abd$ or $adb$. 

By the same reasoning, $\pi(\{a,b,c\})$ and $\pi(\{a,b,e\})$ agree on the position of either $a$ or $b$, and $\pi(\{a,b,d\})$ and $\pi(\{a,b,e\})$ agree on the position of either $a$ or $b$. If $\pi(\{a,b,e\})$ agrees with \emph{either} $\pi(\{a,b,c\})$ or $\pi(\{a,b,d\})$ on the position of $a$, then it agrees with \emph{both} (in which case we are done) since $\pi(\{a,b,c\})$ and $\pi(\{a,b,d\})$ agree on the position of $a$, by the previous paragraph. Thus, the only option is that $\pi(\{a,b,e\})$ agrees with \emph{both} $\pi(\{a,b,c\})$ and $\pi(\{a,b,d\})$ on the position of $b$. This completes the proof.
\end{proof}

%Note that $\{a,b\}$ was chosen without loss of generality so Lemma~\ref{lemma:fix} holds for any arbitrary size two subset of $[k]$. 

\subsubsection*{Step 2: Structure of size $n-2$ sets}

Since $n-2=1$, we begin by fixing a single element $x\in\{a,b,c,d,e\}$. In the following lemma we prove that $x$ cannot be frozen to two different indices.

%two characters $x,y\in \{a,b,c,d,e\}$ are $i$-frozen for the same $i\in[n]$.

\begin{lemma}\label{lem:same}
%If two pairs $\{u,x\},\{u,y\}\subseteq\{a,b,c,d,e\}$ each $i$-freeze an element, then both pairs $i$-freeze $u$. % $p$ and $p'$ respectively, then $p=p'$.
% and a pair $P'$ of elements from $\subseteq\{a,b,c,d,e\}$ also $i$-fix an element  then  $p\in P$ and another pair $P'\subseteq\{a,b,c,d,e\}$ $j$-fixes $p$ then $i=j$.

If a pair $\{x,y\}\subseteq\{a,b,c,d,e\}$ $i$-freezes $x$ and a pair $\{x,z\}\subseteq\{a,b,c,d,e\}$ $j$-freezes $x$ then $i=j$.
%If an element of $[k]$ is fixed by multiple then it is fix to same $i$.
%Each element of $[k]$ can only be $i$-fixed for one value of $i$. 

\end{lemma}

\begin{proof}
Since $\{x,y\}$ $i$-freezes $x$, then in particular, $\pi(\{x,y,z\})[i]=x$. Since $\{x,z\}$ $j$-freezes $x$, then in particular, $\pi(\{x,y,z\})[j]=x$. A single character cannot be in multiple positions of the permutation $\pi(\{x,y,z\})$ so $i=j$.
\end{proof}

\subsubsection*{Step 3: Counting argument}
%For our settings of the parameters, it would be overkill explicitly 

%We will derive a contradiction by showing that steps 1 and 2 are incompatible. Step 1 implies that overall, \emph{many} characters are frozen. On the other hand, step 2 implies a \emph{limit} on the number of frozen characters by saying that a single character cannot be frozen to two different indices. %In this step we show that steps 1 and 2 are incompatible.

Lemma~\ref{lem:fix} implies that for each character $x\in \{a,b,c,d,e\}$ except for at most one, \emph{some} pair $\{x,y\}$ freezes $x$.  That is, at least 4 characters are frozen by some pair. However $n=3$ so by the pigeonhole principle, two characters $x,y\in \{a,b,c,d,e\}$ are frozen to the same index $i$.

Fix $x$, $y$, and $i$, and suppose $x$ and $y$ are each $i$-frozen. By Lemma~\ref{lem:fix}, the pair $\{x,y\}$ freezes either $x$ or $y$. Without loss of generality, say $\{x,y\}$ freezes $x$. By Lemma~\ref{lem:same}, since $x$ is $i$-frozen by some pair, \emph{all} pairs that freeze $x$ must $i$-freeze $x$. Thus, the pair $\{x,y\}$ $i$-freezes $x$.

Let $\{y,z\}\subseteq \{a,b,c,d,e\}$ be a pair that $i$-freezes $y$. Thus we have $\pi(\{x,y,z\})[i]=y$. However, since $\{x,y\}$ $i$-freezes $x$, we also have $\pi(\{x,y,z\})[i]=x$.  This is a contradiction since $\pi(\{x,y,z\})[i]$ cannot take on two different values.
\end{proof}

\subsection{Proof of Theorem~\ref{thm:4r}}\label{app:n<k}

\begin{theorem}[Restatement of Theorem~\ref{thm:4r}]
There exist $n$ and $k$ so that every function $\pi_{n,k}$ has maximum distortion at least 4.
%Let $n\geq $, let $k\geq $, and let $\mathcal{S}_{n,k}$ be the set of all size $n$ multisets of $[k]$. For every function $\pi$ that maps every multiset $S\in \mathcal{S}_{n,k}$ to a permutation of $S$ there exists a pair $S,S'\in \mathcal{S}_{n,k}$ with $S\oplus S'=2$ and $d(\pi(S),\pi(S'))\geq 4$?
\end{theorem}

More specifically, we will show that there exists a constant $c$ so that Theorem~\ref{thm:4} holds for $n\geq 5$ and $k\geq t_{n-1}(cn)$ where the tower function $t_j(x)$ is defined by $t_1(x) = x$ and $t_{i+1}(x) = 2^{t_i(x)}$.

\begin{proof}
%todo: does the wording of the above thm make sense? Maybe instead define function space curly F_{n,k} to mean the set of all functions from multiset to permutation of multiset (or just say set). Then say for all pi in F.
%Let $n\geq 5$ and let $k\geq r_{n-1}(n+7,n+7)$

Suppose by way of contradiction that there is a function $\pi_{n,k}$ with maximum distortion at most 3, for $n$ and $k$ to be set later. 

For the remainder of this section we omit the subscript of $\pi$ since $n$ and $k$ are fixed. As a convention, we will generally use the variables $P$, $Q$, $R$, and $S$ to refer to subsets of $[k]$ of size $n-3$, $n-2$, $n-1$, and $n$, respectively.

\subsubsection{Step 1: Structure of size $n-1$ sets.}

Let $R\subset [k]$ be a size $n-1$ set. Recall from Section~\ref{sec:frame} that $\mathcal{U}_R$ is the set of all size $n$ sets $S$ such that $R\subset S \subset [k]$. We note that all pairs $S,S'\in \mathcal{U}_R$ are by definition such that $|S\oplus S'|=2$. Because we initially supposed that $\pi$ has maximum distortion at most 3, we know that all pairs $S,S'\in\mathcal{U}_R$ have Hamming distance $d(\pi(S),\pi(S'))\leq 3$.

We begin by generalizing the notion of \emph{freezing} a character from Definition~\ref{defn:freeze}. Instead of freezing a single character, our new definition will concern freezing a \emph{set} of characters. Freezing a set of characters essentially means that every character in the set is frozen to a different index.

\begin{definition}[freeze]
Let $R\subseteq [k]$ be a size $n-1$ set and let $A_R\subseteq R$. We say that $R$ \emph{freezes} $A_R$ with \emph{freezing function} $g_R$ if $g_R$ is a one-to-one mapping from $A_R$ to $[n]$ such that for all $a\in A_R$ and all $S\in\mathcal{U}_R$, we have $\pi(S)[g_R(a)]=a$.
\end{definition}

Unlike in the proof of Theorem~\ref{thm:warm}, it is not true that for any size $n-1$ set $R\subset [k]$, some subset of $R$ must be frozen. Instead, to capture the full structure of permutations with Hamming distance 3, we will need another notion of freezing, which we call \emph{semi-freezing}. In this definition, each character is restricted to \emph{two} indices instead of just one. 
%Also, in this definition we now consider \emph{every} element of $R$ instead of just a subset of $R$.

\begin{definition}[semi-freeze]
Let $R\subset [k] $ be a size $n-1$ set. We say that $R$ is \emph{semi-frozen} with \emph{semi-freezing function} $h_R$ and \emph{wildcard index} $w_R$ if $h_R$ is a one-to-one mapping from $R$ to $[n]$ such that for all $r\in R$, we have that for all $S\in\mathcal{U}_R$ either $\pi(S)[h_R(r)]=r$ or $\pi(S)[w_R]=r$.
\end{definition}

We note that since $g_R$ is a one-to-one mapping and $R$ is of size $n-1$, the \emph{only} index in $[n]$ not mapped to by $g_R$ is the wildcard index $w_R$. We call $w_R$ the wildcard index because $\pi(S)$ could place \emph{any} character from $R$ at index $w_R$. In contrast, for every other index $i$, $\pi(S)$ can only place a single character from $R$ at index $i$, namely the character mapped to $i$ by the function $g_R$.

Our structural lemma for step 1 says that either $R$ freezes a large subset $A\subset R$, or $R$ is semi-frozen. 

\begin{lemma}\label{lem:fix2}
Every set $R\subset [k]$ of size $n-1$, obeys one of the following two configurations: \begin{enumerate}
\item there exists a set $A_R\subset R$ of size $n-3$ such that $R$ freezes $A_R$, or 
\item $R$ is semi-frozen.
\end{enumerate}
\end{lemma} 

Figure~\ref{fig:fix2} shows the structure of permutations that obey each of the two configurations in Lemma~\ref{lem:fix2}. In configuration 1, each character in a \emph{large subset} of $R$ is always mapped to a \emph{single} index. In configuration 2, each character in $R$ is always mapped to one of \emph{two} possible choices. In other words, both configurations enforce a rigid structure but each of them are flexible in a different way. Configuration 1 is flexible in that it does not impose structure on characters not in $A_R$, and rigid in that the characters in $A_R$ are \emph{always} mapped to the same position. On the other hand, configuration 2 is flexible in that it allows each character to map to a choice of two positions, but rigid in that the structure is imposed on \emph{every} character in $R$.

\begin{figure}[h]
%\centering    
\begin{subfigure}[b]{.57\linewidth}%\label{fig:struct1}
\includegraphics[width=60mm]{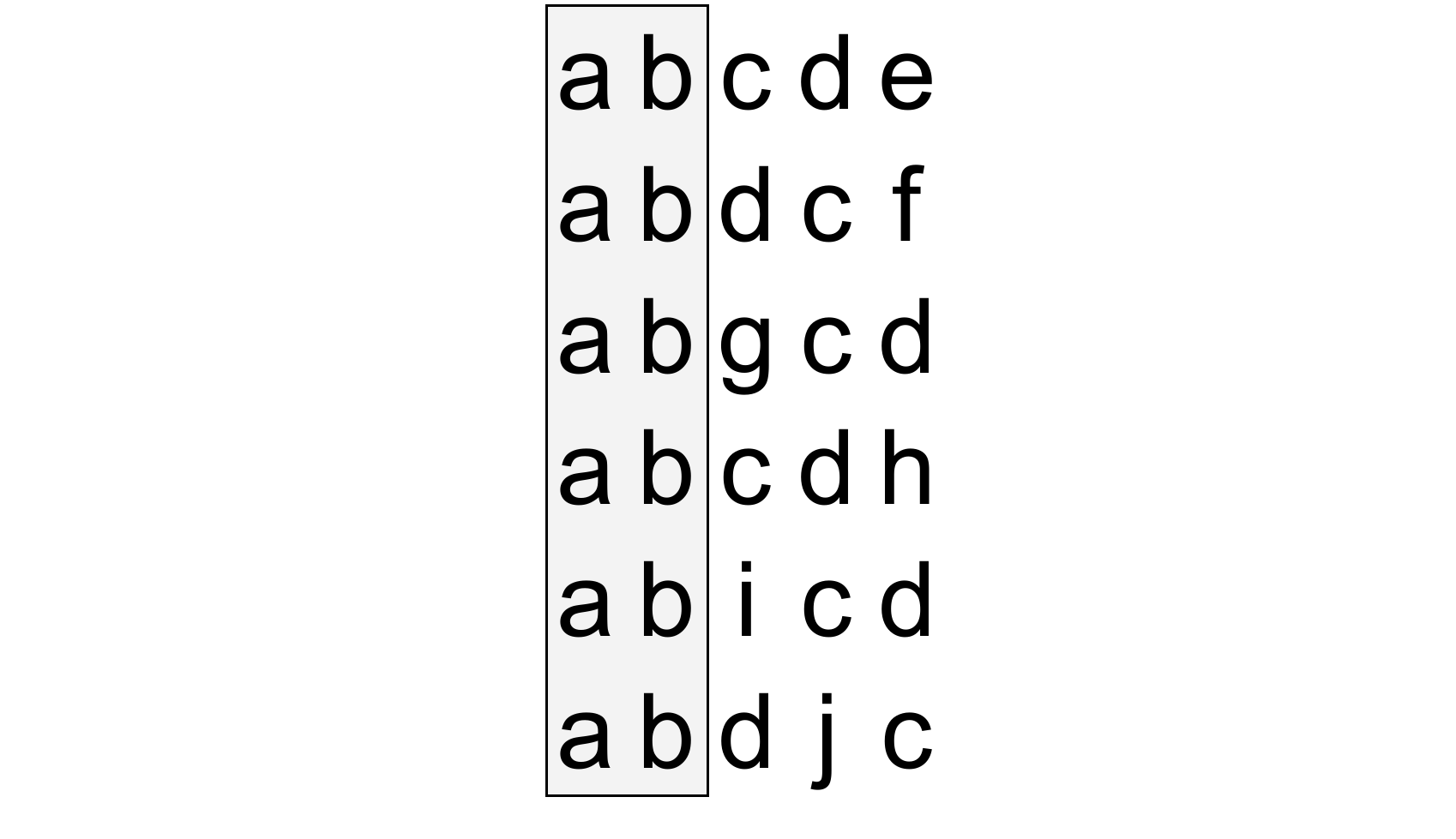}

\end{subfigure}
%\hspace{15mm}
\begin{subfigure}[b]{.15\linewidth}%\label{fig:struct2}
\includegraphics[width=60mm]{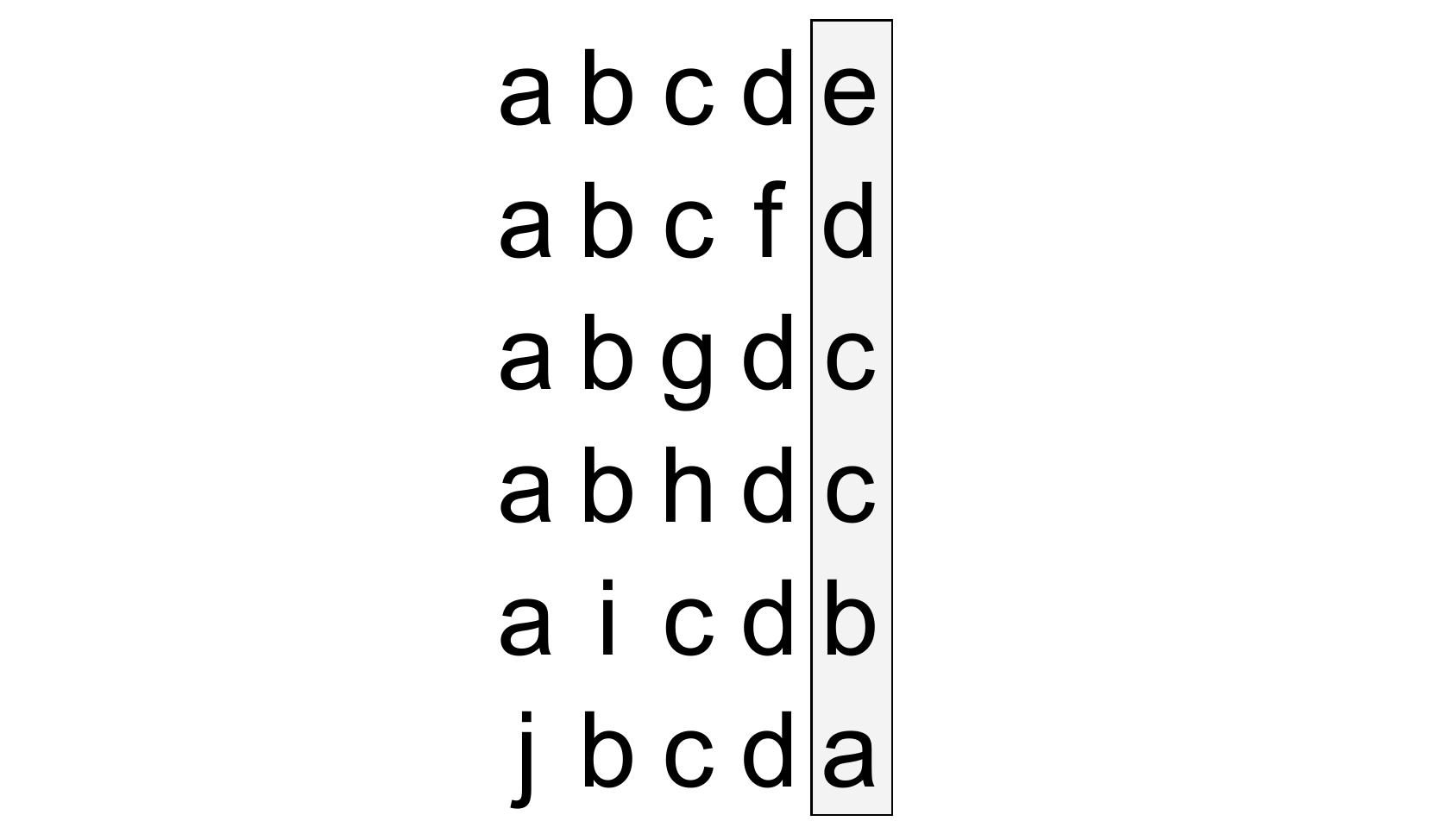}

\end{subfigure}
\caption{\small Examples of the configurations from Lemma~\ref{lem:fix2}. Each of the two subfigures shows the the set of permutations $\pi(S)$ for each $S\in\mathcal{U}_R$ where $R=\{a,b,c,d\}$, $n=5$, and $k=10$. The left subfigure shows configuration 1 of Lemma~\ref{lem:fix2}: the frozen set is $A=\{a,b\}$ since $a$ and $b$ each only appear at a fixed index, as marked by the gray box. The right subfigure shows configuration 2 of Lemma~\ref{lem:fix2}: $R$ is semi-frozen with wildcard index indicated by the gray box since each element of $R$ only appears at the wildcard index and one other index.}
\label{fig:fix2}
\end{figure}

%Recall that we initially assumed that all pairs $S,S'\in\mathcal{U}_R$ have Hamming distance $d(\pi(S),\pi(S'))\leq 3$. For the proof of Lemma~\ref{lem:fix2}, we will prove by contradiction and our goal will be to find a pair $S,S'\in\mathcal{U}_R$ with $d(\pi(S),\pi(S'))> 3$. 
To prove Lemma~\ref{lem:fix2}, we would like to initially fix a pair $S,S'\in\mathcal{U}_R$ with $d(\pi(S),\pi(S'))= 3$. The following lemma proves that we can assume that such a pair exists, because if not, then configuration 1 of Lemma~\ref{lem:fix2} already holds. The proof of the following lemma is nearly identical to the proof of Lemma~\ref{lem:fix}.
%(Actually something a bit stronger.)

%This is the less interesting case, as this was our assumption from the proof of Theorem~\ref{thm:warm}.  Lemma~\ref{lem:fix} is the special case of this statement for when $n=3$ and $k=5$. This proof 
\begin{lemma}\label{lem:exists3}
If all pairs $S,S'\in\mathcal{U}_R$ have $d(\pi(S),\pi(S'))\leq 2$, then for every size $n-1$ set $R\subset [k]$, there exists a subset $A_R$ of size $n-2$ such that $R$ freezes $A_R$.
\end{lemma}

\begin{proof}
Any pair $S,S'\in\mathcal{U}_R$ must agree on the position of at least $n-2$ characters in $R$ because otherwise we would have $d(\pi(S), \pi(S'))>2$. Also, there must exist a pair $S,S'\in\mathcal{U}_R$ with $d(\pi(S), \pi(S'))=2$ because if all pairs had $d(\pi(S), \pi(S'))=1$ then all $S\in \mathcal{U}_R$ would agree on the position of all $n-1$ characters in $R$ and we would be done. Thus, let $S,S'\in\mathcal{U}_R$ be such that $d(\pi(S), \pi(S'))=2$.

Let $R=\{a_1,a_2,\dots,a_{n-1}\}$ and without loss of generality suppose $\pi(S)$ and $\pi(S')$ agree on the position of the characters $a_1,a_2,\dots,a_{n-2}$ and disagree on the position of $a_{n-1}$. We wish to show that for all $S''\in\mathcal{U}_R$, $\pi(S'')$ also agrees with $\pi(S)$ and $\pi(S')$ on the position of the characters $a_1,a_2,\dots,a_{n-2}$. Suppose for contradiction that there exists $S''\in\mathcal{U}_R$ such that $\pi(S'')$ disagrees with $\pi(S)$ and $\pi(S')$ on the position of some character in $\{a_1,a_2,\dots,a_{n-2}\}$. Then since $\pi(S'')$ must agree with both $\pi(S)$ and $\pi(S')$ on the position of at least $n-2$ characters in $R$, $\pi(S'')$ must agree with \emph{both} $\pi(S)$ and $\pi(S')$ on the position of $a_{n-1}$. But,  this is a contradiction because $\pi(S)$ and $\pi(S')$ disagree on the position of $a_{n-1}$.
%This means that $\pi(S)$ and $\pi(S')$ agree on the position of all $n-1$ characters in $R$, which contradicts our original supposition that  $d(\pi(S), \pi(S'))=2$.
\end{proof}

\begin{proof}[Proof of Lemma~\ref{lem:fix2}]
Let $S,S'\in\mathcal{U}_R$ be such that $d(\pi(S),\pi(S'))=3$. Such $S,S'$ exist by Lemma~\ref{lem:exists3}. Fix $S,S'\in\mathcal{U}_R$. Let $R=\{a_1,a_2,\dots,a_{n-1}\}$ and without loss of generality suppose $\pi(S)$ and $\pi(S')$ agree on the position of the $n-3$ characters $a_3,a_4,\dots,a_{n-1}$. If every $S''\in\mathcal{U}_R$ is such that $\pi(S'')$ also agrees with $\pi(S)$ and $\pi(S')$ on the positions of the characters $a_3,a_4,\dots,a_{n-1}$, then we are done because in this case $R$ freezes the set $\{a_3,a_4,\dots,a_{n-1}\}$. So suppose otherwise; that is, let $S''\in\mathcal{U}_R$ be such that $\pi(S'')$ disagrees with $\pi(S)$ and $\pi(S')$ on the position of $a_3$ (without loss of generality). 

Since $S,S',S''\in \mathcal{U}_R$, each of $S$, $S'$, and $S''$ have one additional character besides those in $R$. Let $s$, $s'$, and $s''$ be these characters respectively.
%Let $\{s\}=S\setminus R$, $\{s'\}=S'\setminus R$, and $\{s''\}=S''\setminus R$. We note that each of these expressions is indeed one single character since $|S|=|S'|=|S''|=n$, $|R|=n-1$, and $R$ is a subset of $S$, $S'$, and $S''$.
In the following we will analyze $\pi(S)$, $\pi(S')$ and $\pi(S'')$. Since $a_4,a_5,\dots a_{n-1}$ are all in the same position with respect to all three permutations, we will ignore these characters. That is, letting $Z=\{a_4,a_5,\dots a_{n-1}\}$, we consider $S\setminus Z=\{a_1,a_2,a_3,s\}$, $S'\setminus Z=\{a_1,a_2,a_3,s'\}$, and $S''\setminus Z=\{a_1,a_2,a_3,s''\}$. We will abuse notation and let $\pi(S\setminus Z)$ be the subpermutation of $\pi(S)$ containing only the elements of $S\setminus Z$, and similarly for $S'\setminus Z$ and $S''\setminus Z$. 

Suppose without loss of generality that $\pi(S\setminus Z)=a_1a_2sa_3$. Then since $\pi(S)$ and $\pi(S')$ agree on the position of $a_3$ but disagree on the positions of $a_1$ and $a_2$, we have that either \begin{enumerate}
\item $\pi(S')$ places $s'$ in the position that $\pi(S)$ places $s$, so $\pi(S'\setminus Z)=a_2a_1s'a_3$, or 
\item $\pi(S')$ places $s'$ in the position that $\pi(S)$ places $a_1$ or $a_2$, so without loss of generality $\pi(S'\setminus Z)=s'a_1a_2a_3$.
\end{enumerate}

Recall that $\pi(S'')$ disagrees with $\pi(S)$ on the position of $a_3$. Since $d(\pi(S''),\pi(S))\leq 3$ and the positions of $a_3$ and $s''$ each account for one unit of difference between $\pi(S'')$ and $\pi(S)$, we know that $\pi(S'')$ agrees with $\pi(S)$ on the position of at least one of $a_1$ or $a_2$. Similarly, since $\pi(S'')$ disagrees with $\pi(S')$ on the position of $a_3$, we have that $\pi(S'')$ agrees with $\pi(S')$ on the position of at least one of $a_1$ or $a_2$. If $\pi(S'\setminus Z)=a_2a_1s'a_3$ (case 1 above), then $\pi(S'')$ cannot possibly agree with both $\pi(S)$ and $\pi(S')$ on the position of at least one of $a_1$ or $a_2$ because the positions of $a_1$ and $a_2$ are swapped in $\pi(S)$ as compared to $\pi(S')$. Thus, it must be the case that $\pi(S'\setminus Z)=s'a_1a_2a_3$ (case 2 above). 

Now, given that $\pi(S\setminus Z)=a_1a_2sa_3$ and $\pi(S'\setminus Z)=s'a_1a_2a_3$, there is only one possibility for $\pi(S''\setminus Z)$ that satisfies the criteria that $\pi(S'')$ disagrees with $\pi(S)$ and $\pi(S')$ on the position of $a_3$ and agrees with each of $\pi(S)$ and $\pi(S')$ on the position of at least one of $a_1$ or $a_2$. The only possibility is that $\pi(S''\setminus Z)=a_1a_3a_2s''$.

The above argument applies for any $S''\in\mathcal{U}_R$ such that $\pi(S'')$ disagrees with $\pi(S)$ and $\pi(S')$ on the position of some $a_i$ with $3\leq i\leq n-1$. That is, letting $Z'=\{a_3,a_4,\dots a_{n-1}\}\setminus\{a_i\}$ and letting $s''=S''\setminus R$, we have that without loss of generality, $\pi(S\setminus Z')=a_1a_2sa_i$, $\pi(S'\setminus Z')=s'a_1a_2a_i$, and $\pi(S''\setminus Z)=a_1a_ia_2s''$. 

We claim that the structure we have derived implies that $R$ is semi-frozen. To see this, consider the following semi-freezing function $h_R$: \begin{align*}&h_R(a_1)=\pi(S)[a_1],\\ &h_R(a_2)=\pi(S')[a_2],\\ \text{for each element $a_i$ for $3\leq i\leq n-1$, } &h_R(a_i)=\pi(S)[a_i]=\pi(S')[a_i],\\ \text{and the wildcard index } &w_R=\pi(S)[a_2]=\pi(S')[a_1].
\end{align*}
\end{proof}

\subsubsection{Treating configurations 1 and 2 independently}\label{sec:ramsey}
Before moving to step 2 of the proof framework, we will show using Ramsey theory that it suffices to consider each of the two configurations from Lemma~\ref{lem:fix2} \emph{independently}. Lemma~\ref{lem:fix2} shows that every size $n-1$ subset of $[k]$ obeys one of two configurations. Using Ramsey theory, we will show that there must exist a subset $K'\subseteq [k]$ such that either \emph{all} size $n-1$ subsets of $K'$ obey configuration 1 or \emph{all} size $n-1$ subsets of $K'$ obey configuration 2. This will allow us to avoid reasoning about the complicated interactions between the two configurations.

The required size $k'=|K'|$ can be expressed as a \emph{hypergraph Ramsey number}.  The hypergraph Ramsey number $r_j(t,t)$ is the minimum value $m$ such that every red-blue coloring of the $j$-tuples of an $m$-element set contains either a red set or a blue set of size $t$, where a set is called red
(blue) if all $j$-tuples from this set are red (blue). Thus, it suffices to let $k'$ satisfy $r_{n-1}(k',k')=k$.

Erd\H{o}s and Rado~\cite{erdos1952combinatorial} give the following bound on $r_j(t,t)$, as stated in~\cite{conlon2010hypergraph}.
There exists a constant $c$ such that: 
\[r_{n-1}(k',k')\leq t_{n-1}(ck')\] where the tower function $t_j(x)$ is defined by $t_1(x) = x$ and $t_{i+1}(x) = 2^{t_i(x)}$.

In the following we will show that it suffices to let $n\geq 5$ and $k'\geq n+7$, so it suffices to set $n\geq 5$ and $k\geq t_{n-1}(cn)$.

\subsubsection{Step 2a: Structure of size $n-2$ sets for configuration 1}
From the previous section, there exists a size $k'$ set $K'$ of tasks such that either  all size $n-1$ subsets of $K'$ obey configuration 1 or all size $n-1$ subsets of $K'$ obey configuration 2. In this section we will assume that all size $n-1$ subsets of $K'$ obey configuration 1, and later we will independently consider configuration 2. 

Recall that configuration 1 says that for every set $R\subset[k']$ of size $n-1$, there exists a set $A_R\subset R$ of size $n-3$ such that $R$ freezes $A_R$. Recall that $g_R$ is the freezing function.

Let $Q\subset [k']$ be a size $n-2$ set. Recall that $\mathcal{U}_Q$ is the set of all size $n-1$ sets $R$ such that $Q\subset R\subset [k']$. We will prove the following simple structural lemma, analogous to Lemma~\ref{lem:same}, which says that the freezing functions for any two sets in $\mathcal{U}_Q$ are \emph{consistent}. 

\begin{lemma}\label{lem:config1}
For every size $n-2$ set $Q\subset [k']$, for any pair $R,R'\in \mathcal{U}_Q$, for any $a\in A_R\cap A_{R'}$, $g_R(a)=g_{R'}(a)$.
\end{lemma}

\begin{proof}
Fix $R,R'\in \mathcal{U}_Q$. Since $R$ and $R'$ are each composed by adding a single character to $Q$, we have that $|R\cup R'|=n$ and $R\cup R'\in \mathcal{U}_R\cap \mathcal{U}_{R'}$. 
%Let $r$ and $r'$ be these characters, respectively. Now, consider the size $n$ set $S=Q\cup\{r\}\cup\{r'\}$. We observe that $S\in \mathcal{U}_R\cap \mathcal{U}_{R'}$.

Since $R\cup R'\in \mathcal{U}_R$, we know that $a$ is at position $g_R(a)$ in $\pi(R\cup R')$ and since $R\cup R'\in \mathcal{U}_{R'}$, we know that $a$ is at position $g_{R'}(a)$ in $\pi(R\cup R')$. Then, since $a$ can only occupy a single position in the permutation $\pi(R\cup R')$, we have that $g_R(a)=g_{R'}(a)$.
\end{proof}

\subsubsection{Step 3a: Counting argument for configuration 1}
Like the previous section, in this section we will assume that all size $n-1$ subsets of $K'$ obey configuration 1. 
%Like step 3 of Theorem~\ref{thm:warm}, we will derive a contradiction by showing that steps 1 and 2 are incompatible. Step 1 implies that overall, \emph{many} characters are frozen. On the other hand, step 2 implies a \emph{limit} on the number of frozen characters by saying that a single character cannot be frozen to two different indices. However, 
Unlike step 3 of Theorem~\ref{thm:warm}, it does not suffice to simply show that two characters are frozen to the same index. Instead, we apply a more sophisticated counting argument.

By Lemma~\ref{lem:config1}, for every size $n-2$ set $Q\subset [k']$, we have that all $R\in\mathcal{U}_Q$ agree on the value of $g_R(a)$ if it exists. Thus, we can define $G_Q$ as the union of $g_R$s over all $R\in\mathcal{U}_Q$. Formally, for any $a\in [k']$, $G_Q(a)=i$ if for every $R\in \mathcal{U}_Q$ with $a\in A_R$, we have $g_R(a)=i$. We note that $G_Q(a)$ \emph{exists} if for \emph{some} $R\in \mathcal{U}_Q$, $a\in A_R$.

Since $|Q|=n-2$ and there are $n$ indices total, $G_Q(a)$ can exist for at most $n-2$ characters $a\in Q$ and at most 2 characters $a\not\in Q$. We say that the pair $(Q,a)$ is \emph{irregular} if $G_Q(a)$ exists and $a\not\in Q$. The quantity that we will count is the total number of irregular pairs $(Q,a)$ over all size $n-2$ sets $Q\subset [k']$ and all $a\in [k']$.

On one hand, as previously mentioned, each set $Q$ can only be in at most 2 irregular pairs. Then since there are ${k'\choose n-2}$ sets $Q\subset [k']$ of size $n-2$, the total number of irregular pairs is at most $2{k'\choose n-2}$.

On the other hand, the definition of configuration 1 implies a lower bound on the number of irregular pairs. Recall that configuration 1 says that for every size $n-1$ set $R\subset [k']$, there exists a set $A_R\subset R$ of size $n-3$ such that $R$ freezes $A_R$. Fix sets $R$ and $A_R$. We claim that for each $a\in A_R$, the pair $(R\setminus\{a\},a)$ is an irregular pair. Firstly,  is clear that $a\not\in R\setminus\{a\}$. Secondly, $G_{R\setminus\{a\}}(a)$ exists because $R\in \mathcal{U}_{R\setminus\{a\}}$ and $a\in A_R$. Thus, for each $a\in A_R$, the pair $(R\setminus\{a\},a)$ is an irregular pair.
 
 Thus, every size $n-1$ set $R\subset [k']$ \emph{produces} $n-3$ irregular pairs $(R\setminus\{a\},a)$. Furthermore, given an irregular pair $(Q,a)$, there is only one set that could produce it, namely $Q\cup\{a\}$. Then since there are ${k'\choose n-1}$ sets $R\subset [k']$ of size $n-1$, we have that the total number of irregular pairs is at least $(n-3){k'\choose n-1}$.
 
 Thus, we have shown that the total number of irregular pairs is at most $2{k'\choose n-2}$ and at least $(n-3){k'\choose n-1}$. Therefore, we have reached a contradiction if $2{k'\choose n-2}<(n-3){k'\choose n-1}$ which is true if $n\geq 4$ and $k'>\frac{n^2-3n+4}{n-3}$. In particular, $n\geq 5$, $k'\geq n+7$ satisfy these bounds.

\subsubsection{Step 2b: Structure of size $n-2$ sets for configuration 2}

From Section~\ref{sec:ramsey}, there exists a size $k'$ set $K'$ of tasks such that either  all size $n-1$ subsets of $K'$ obey configuration 1 or all size $n-1$ subsets of $K'$ obey configuration 2. We have already considered the configuration 1 case and now we will assume that all size $n-1$ subsets of $K'$ obey configuration 2. Recall that configuration 2 says that $R$ is semi-frozen. Recall that $h_R$ is the semi-freezing function and $w_R$ is the wildcard index.

Let $Q\subset [k']$ be a size $n-2$ set. Recall that $\mathcal{U}_Q$ is the set of all size $n-1$ sets $R$ such that $Q\subset R\subset [k']$. We will prove the following structural lemma, which says that the semi-freezing functions for two sets in $\mathcal{U}_Q$ are in some sense \emph{consistent}. 

\begin{lemma}\label{lem:config2}
For every size $n-2$ set $Q\subset [k']$, there exists a size $n-4$ set $_QT\subset Q$ such that for all $R,R'\in \mathcal{U}_Q$ and all $t\in T_Q$, $h_R(t)=h_{R'}(t)$.
\end{lemma}

We will prove Lemma~\ref{lem:config2} through a series of lemmas. In the following lemma, we consider the characters that are \emph{not} in the set $T_Q$, that is, the characters $q\in Q$ for which $h_R(q)\not=h_{R'}(q)$.

\begin{lemma}\label{lem:wild}
For all size $n-2$ sets $Q\subset [k']$ and all $R,R'\in \mathcal{U}_Q$, if $q\in Q$ is such that $h_R(q)\not=h_{R'}(q)$, then either $\pi(R\cup R')[w_R]=q$ or $\pi(R\cup R')[w_{R'}]=q$.
\end{lemma}

\begin{proof}
  Since $R$ and $R'$ are each composed by adding a single character to $Q$, we have $R\cup R'\in \mathcal{U}_R\cap\mathcal{U}_{R'}$. Since $R\cup R'\in\mathcal{U}_R$, we know that the position of $q$ in $\pi(R\cup R')$ is either $h_R(q)$ or $w_R$. Since $R\cup R'\in\mathcal{U}_{R'}$, we know that the position of $q$ in $\pi(R\cup R')$ is either $h_{R'}(q)$ or $w_{R'}$. Thus, the position of $q$ in $\pi(R\cup R')$ must be either $w_R$ or $w_{R'}$, because otherwise its position would have to be both $h_R(q)$ and $h_{R'}(q)$, which cannot happen since $h_R(q)\not=h_{R'}(q)$.
\end{proof}

Before proving Lemma~\ref{lem:config2}, we prove the pairwise version of Lemma~\ref{lem:config2}.

\begin{lemma}[pairwise version of Lemma~\ref{lem:config2}]\label{lem:pair}
For every size $n-2$ set $Q\subset [k']$ of characters, for every pair $R,R'\in \mathcal{U}_Q$, there exists a size $n-4$ subset $T\subset Q$ such that every character $t\in T$, $h_R(t)=h_{R'}(t)$.
\end{lemma}

\begin{proof}
Suppose by way of contradiction that there exist $R,R'\in \mathcal{U}_Q$ such that there is a set of 3 characters $Q'\subset Q$ so that for each $q\in Q'$,  $h_R(q)\not=h_{R'}(q)$. By Lemma~\ref{lem:wild}, for each $q\in Q'$ the position of $q$ in $\pi(R\cup R')$ is either $w_R$ or $w_{R'}$. That is, all 3 characters in $Q'$ must occupy a total of 2 positions in $\pi(R\cup R')$, which is impossible.
\end{proof}

%\begin{lemma}
%If $q$ and $q'$ are both such that $h_R(q)\not=h_{R'}(q)$ and $h_R(q')\not=h_{R'}(q')$, then $h_R(q)=h_{R'}(r')$ and $h_{R'}(q')=h_R(r)$. (Or with $q$ and $q'$ switched).
%\end{lemma}

%For the proof of Lemma~\ref{lem:config2}, we will prove by contradiction and our goal will be to show that there exist $R,R'\in \mathcal{U}_Q$ such that there is a set of 3 characters $Q'\in Q$ so that for each $q\in Q'$,  $h_R(q)\not=h_{R'}(q)$. To do this, we would like to initially fix a pair of 2 such characters. The following lemma proves that we can assume that 2 such characters exist, because if not, then Lemma~\ref{lem:config2} already holds. (actually something a bit stronger)

We have just shown in Lemma~\ref{lem:pair} that given a size $n-2$ set $Q\subset [k']$, for every pair $R,R'\in\mathcal{U}_Q$ there are at most two characters $q$ in $Q$ for which $h_R(q)\not=h_{R'}(q)$. Thus, we have two cases: 1) the uninteresting case where every pair $R,R'$ has only one such character $q$, and 2) the interesting case where there exist $R,R'$ so that there are two such characters $q$.  The following lemma handles the uninteresting case by showing that in this case Lemma~\ref{lem:config2} already holds.

%exists a pair $R,R'\in\mathcal{U}_Q$ so that there exist $q,q'\in Q$ with $h_R(q)\not=h_{R'}(q)$ and $h_R(q')\not=h_{R'}(q')$, and 2) the uninteresting case where for every pair $R,R'\in\mathcal{U}_Q$ there is at most one character $q\in Q$ with $h_R(q)\not=h_{R'}(q)$.

\begin{lemma}\label{lem:one}
Suppose $Q\subset [k']$ is a size $n-2$ set of characters and for every pair $R,R'\in \mathcal{U}_Q$ there is at most one character $q\in Q$, with $h_R(q)\not=h_{R'}(q)$. Then there exists a size $n-3$ subset $T\subset Q$ such that for every pair $R,R'\in \mathcal{U}_Q$ and every character $t\in T$, $h_R(t)=h_{R'}(t)$.
\end{lemma}
\begin{proof}
Let $R,R'\in \mathcal{U}_Q$ and $q\in Q$ be such that $h_R(q)\not=h_{R'}(q)$. Then by assumption, for all $q'\in Q$ with $q'\not=q$, $h_R(q')=h_{R'}(q')$. Consider $R''\in \mathcal{U}_Q$. It suffices to show that for all $q'\in Q$ with $q'\not=q$, we have $h_{R''}(q')=h_{R}(q')$. Suppose for contradiction that there exists $q'\in Q$ with $q'\not=q$ such that $h_{R''}(q')\not=h_{R}(q')$. Then, since $h_{R}(q')=h_{R'}(q')$, we have $h_{R''}(q')\not=h_{R'}(q')$. From the precondition of the lemma statement, $q'$ is the only character with $h_{R''}(q')\not=h_{R}(q')$ and $q'$ is the only character with $h_{R''}(q')\not=h_{R'}(q')$. Thus, $h_{R''}(q)=h_{R}(q)$ and $h_{R''}(q)=h_{R'}(q)$. So, $h_{R}(q)=h_{R'}(q)$, a contradiction.
\end{proof}

We have handled the uninteresting case from above and now we handle the interesting case in which there exist $R,R'\in \mathcal{U}_Q$ so that there are two characters $q$ in $Q$ for which $h_R(q)\not=h_{R'}(q)$. The following lemma shows that in this case we can completely characterize the structure of $h_R$ and $h_{R'}$. Table~\ref{tabh} depicts the structure.

\begin{lemma}\label{lem:spec}
For every size $n-2$ set $Q\subset [k']$ of characters, if $R,R'\in \mathcal{U}_Q$ are such that there exist $q,q'\in Q$ with $h_R(q)\not=h_{R'}(q)$ and $h_R(q')\not=h_{R'}(q')$, then (modulo switching $q$ and $q'$):
\begin{enumerate}
\item Let $r$ be the single character in $R\setminus Q$ and let $r'$ be the single character in $R'\setminus Q$. Then $h_R(q)=h_{R'}(r')$ and $h_R(r)=h_{R'}(q')$.
\item $h_R(q')=w_{R'}$ and $h_{R'}(q)=w_R$.
\item for all $q''\in Q$ not equal to $q$ or $q'$,  $h_R(q'')=h_{R'}(q'')$.
\end{enumerate}
\end{lemma}
%We note that it could also be the case that these structural items hold with $q$ and $q'$ switched (say without loss of generality?)

\begin{table}[h]
\centering
\begin{tabular}{c|c|c|c|c|c}
& &  &$\bm{w_R}$  &$\bm{w_{R'}}$  &  \\
\hline
 $\bm{\pi(R\cup R')}$& $r'$ & $r$ & $q$ & $q'$ &$\{q''\in Q\}$\\
 \hline
$\bm{h_R}$ & $q$  & $r$ & N/A & $q'$ &$\{q''\in Q\}$\\
\hline
 $\bm{h_{R'}}$ & $r'$ & $q'$ & $q$ & N/A &$\{q''\in Q\}$
\end{tabular}
\caption{The structure imposed by Lemma~\ref{lem:spec}. Each column indicates an index in $[n]$. For example, the first column indicates that $\pi(R\cup R')[r']=h_R(q)=h_{R'}(r')$. Some entries are not applicable (N/A) because by definition $h_R$ does not map anything to the index $w_R$.}\label{tabh}
\end{table}

\begin{proof}
We begin with item 3. By Lemma~\ref{lem:wild}, without loss of generality $\pi(R\cup R')[w_R]=q$ and $\pi(R\cup R')[w_{R'}]=q'$. Since the position in $\pi(R\cup R')$ of each remaining character $q''\in Q$ is either $h_R(q'')$ or $w_R$ but the position $w_R$ is taken by $q$, it must be that $\pi(R\cup R')[h_R(q'')]=q''$. Similarly, we have $\pi(R\cup R')[h_{R'}(q'')]=q''$. Thus, $h_R(q'')=h_{R'}(q'')$ for all $q''\in Q$ with $q''\not=q,q'$.

We now move to item 2. Since the position of $q$ in $\pi(R\cup R')$ is either $w_{R'}$ or $h_{R'}(q)$ but $w_{R'}$ is taken by $q'$, we have that $\pi(R\cup R')[h_{R'}(q)]=q$. We already know that $\pi(R\cup R')[w_R]=q$, so $h_{R'}(q)=w_R$. By a symmetric argument, $h_R(q')=w_{R'}$.

We now move to item 1. Since the position in $\pi(R\cup R')$ of $r$ is either $h_R(r)$ or $w_R$, but the position $w_R$ is taken by $q$, it must be that $\pi(R\cup R')[h_R(r)]=r$. 
%Similarly, since the position in $\pi(R\cup R')$ of $r'$ is either $h_{R'}(r')$ or $w_{R'}$, but the position $w_{R'}$ is taken by $q'$, it must be that $\pi(R\cup R')[h_{R'}(r')]=r'$. 
By a symmetric argument, $\pi(R\cup R')[h_{R'}(r')]=r'$.
Combining these two facts, since $r$ and $r'$ cannot occupy the same index in $\pi(R\cup R')$, we have $h_R(r)\not=h_{R'}(r')$. We proceed by process of elimination.

The set of indices which are mapped to by $h_R$ is $[n]\setminus \{w_R\}$ and the indices which have so far been mapped to by items 2 and 3 are $w_{R'}$ and $h_R(q'')=h_{R'}(q'')$ for all $q''\in Q$ not equal to $q$ or $q'$. The set of indices which are mapped to by $h_{R'}$ is $[n]\setminus \{w_{R'}\}$ and the indices which have so far been mapped to by items 2 and 3 are $w_{R}$ and $h_{R'}(q'')=h_{R}(q'')$ for all $q''\in Q$ with $q''\not=q,q'$. Thus, the set of indices which have not yet been mapped to is the same for $h_R$ and $h_{R'}$: $[n]\setminus\{w_R,w_{R'},h_R(q'')=h_{R'}(q'')\}$. The characters for which $h_R$ has not yet been determined are $q$ and $r$ and the characters for which $h_{R'}$ has not yet been determined are $q'$ and $r'$. From the previous paragraph, we know that $h_R(r)\not=h_{R'}(r')$. Thus, we have $h_R(q)=h_{R'}(r')$ and $h_R(r)=h_{R'}(q')$.
\end{proof}

We are now ready to prove Lemma~\ref{lem:config2}. We have just shown in Lemma~\ref{lem:spec} that individual pairs $R,R'\in \mathcal{U}_Q$ obey a particular structure, and in Lemma~\ref{lem:config2} we will derive structure among \emph{all} $R\in\mathcal{U}_Q$. 

%which So far we have proved structural lemmas about individual pairs $R,R'\in \mathcal{U}_Q$, but now we wish to show that \emph{all} elements $R$ of $\mathcal{U}_Q$ agree on $h_R$ over a large set of inputs.

\begin{lemma}[Restatement of Lemma~\ref{lem:config2}]\label{lem:config22}
For every size $n-2$ set $Q\subset [k']$, there exists a size $n-4$ set $T_Q\subset Q$ such that for all $R,R'\in \mathcal{U}_Q$ and all $t\in T_Q$, $h_R(t)=h_{R'}(t)$.
\end{lemma}

\begin{proof}
By Lemma~\ref{lem:one} we can assume that there exists a pair $R,R'\in \mathcal{U}_Q$ so that there exist $q,q'\in Q$ with $h_R(q)\not=h_{R'}(q)$ and $h_R(q')\not=h_{R'}(q')$. Fix $R$, $R'$, $q$, and $q'$. $R$ and $R'$ obey the structure specified by Lemma~\ref{lem:spec}.
Consider $R''\in \mathcal{U}_Q$ with $R''\not=R,R'$. Suppose by way of contradiction that there exists $q''\in Q$ with $q''\not=q,q'$ such that $h_{R''}(q'')\not=h_{R}(q'')$ (and thus also $h_{R''}(q'')\not=h_{R'}(q'')$ since $h_{R}(q'')=h_{R'}(q'')$ by Lemma~\ref{lem:spec}).

We first note that it cannot be the case that both $h_{R''}(q)=h_{R}(q)$ and $h_{R''}(q')=h_{R}(q')$ because then we would have $h_{R''}(q)\not=h_{R'}(q)$ and $h_{R''}(q')\not=h_{R'}(q')$, in which case $h_{R''}$ and $h_{R'}$ would differ on inputs $q$, $q'$, and $q''$ which contradicts Lemma~\ref{lem:pair}. Thus, $h_{R''}$ must differ from each of $h_R$ and $h_{R'}$ and on exactly one of $q$ or $q'$. Without loss of generality, suppose $h_{R''}(q)\not=h_R(q)$ and $h_{R''}(q')\not=h_{R'}(q')$.

Since $h_{R''}$ also differs from each of $h_R$ and $h_{R'}$ on input $q''$, $h_{R''}$ differs from each of $h_R$ and $h_{R'}$ on exactly two inputs. Thus, the pair $R'',R$ and the pair $R'',R'$ both obey the structure specified by Lemma~\ref{lem:spec}. 
%We have shown that all three pairs from $\{R,R',R''\}$ obey the structure specified in Lemma~\ref{lem:spec}. 
We claim that it is impossible to reconcile these pairwise structural constraints.

$R$, $R'$, and $R''$ are each composed by adding a single character to $Q$. Let $r$, $r'$, and $r''$ be these characters respectively. Applying item 1 of Lemma~\ref{lem:spec} to the pair $R,R'$ we have the following two cases:

\subparagraph{Case 1: $h_R(q)=h_{R'}(r')$ and $h_R(r)=h_{R'}(q')$.} Item 1 of Lemma~\ref{lem:spec} presents two options for the pair $R,R''$: either $h_{R''}(r'')=h_R(q)$ or $h_{R''}(r'')=h_R(q'')$. If $h_{R''}(r'')=h_R(q)$, then from the definition of case 1, $h_{R'}(r')=h_{R''}(r'')$, but this is not true by item 1 of Lemma~\ref{lem:spec}. Thus, $h_{R''}(r'')=h_R(q'')$ and $h_R(r)=h_{R''}(q)$. Since $h_{R''}$ and $h_{R'}$ differ only on inputs $q'$ and $q''$, we have $h_{R''}(q)=h_{R'}(q)$. Thus, we have shown that $h_R(r)=h_{R'}(q)$. However, by item 2 of Lemma~\ref{lem:spec}, we have $h_{R'}(q)=w_R$, which is a contradiction since $h_R(r)\not=w_R$.

\subparagraph{Case 2: $h_R(q')=h_{R'}(r')$ and $h_R(r)=h_{R'}(q)$.} Since $h_{R''}$ and $h_{R}$ differ only on inputs $q$ and $q''$, we have $h_{R''}(q')=h_{R}(q')$. Thus, $h_{R''}(q')=h_{R'}(r')$. Then by item 2 of Lemma~\ref{lem:spec}, we have $h_{R''}(q'')=w_{R'}$.
Since $h_{R''}$ and $h_{R'}$ differ only on inputs $q'$ and $q''$, we have $h_{R''}(q)=h_{R'}(q)$. Then since we are in case 2, we have $h_R(r)=h_{R''}(q)$. Then by item 2 of Lemma~\ref{lem:spec}, we have $h_{R''}(q'')=w_{R}$. Thus, we have shown that $h_{R''}(q'')$ is equal to both $w_{R'}$ and $w_R$, which is not true by Lemma~\ref{lem:spec}.
\end{proof}

By Lemma~\ref{lem:config22}, 
%for every size $n-2$ set $Q\subset [k']$, there exists a size $n-4$ subset $T_Q\subset Q$ such that for every pair $R,R'\in \mathcal{U}_Q$ and every character $t\in T_Q$, $h_R(t)=h_{R'}(t)$. Thus, 
we can define a function $h'_Q$ that takes as input any element $t\in T_Q$ and outputs the value $h_R(t)$, which is the same for all $R\in \mathcal{U}_Q$.

Let $P\subset [k']$ be a size $n-3$ set. Recall that $\mathcal{U}_P$ is the set of all size $n-2$ sets $Q$ such that $P\subset Q \subset [k']$. We conclude this section by proving a lemma similar to Lemma~\ref{lem:config1}, which says that the functions $h'$ for any two sets in $\mathcal{U}_P$ are \emph{consistent}.

\begin{lemma}\label{lem:config2con}
For every size $n-3$ set $P\subset [k']$, for any pair $Q,Q'\in \mathcal{U}_P$, for any character $t\in T_Q\cap T_{Q'}$, $h'_Q(t)=h'_{Q'}(t)$.
\end{lemma}

\begin{proof}
%Fix $Q,Q'\in \mathcal{U}_{Q'}$. By definition, 
Since $Q$ and $Q'$ are each composed by adding a single character to $P$, we have $Q\cup Q'\in\mathcal{U}_Q\cap \mathcal{U}_{Q'}$. Since $Q\cup Q'\in \mathcal{U}_Q$, we know that $h'_Q(t)=h_{Q\cup Q'}(t)$ and since $Q\cup Q'\in \mathcal{U}_{Q'}$, we know that $h'_{Q'}(t)=h_{Q\cup Q'}(t)$. Thus, $h'_Q(t)=h'_{Q'}(t)$.
\end{proof}

\subsubsection{Step 3b: Counting argument for configuration 2}

Like the previous section, in this section we will assume that all size $n-1$ subsets of $K'$ obey configuration 2. The counting argument similar to that from step 3a.

%In this section we will perform step 3 of the proof framework under the assumption that all size $n-1$ subsets of $[k']$ obey configuration 2 from Lemma~\ref{lem:fix2}. The counting argument will be similar to the step 3a, the counting argument for configuration 1.

By Lemma~\ref{lem:config2con}, for every size $n-3$ set $P\subset [k']$, we have that all $Q\in\mathcal{U}_P$ agree on the value of $h'_Q(t)$ if it exists. Thus, we can define $H_P$ as the union of $h'_Q$s over all $Q\in\mathcal{U}_P$. Formally, $H_P(t)=i$ if for every $Q\in \mathcal{U}_P$ with $t\in T_Q$, we have $h'_Q(t)=i$. We note that $H_P(t)$ \emph{exists} if for \emph{some} $Q\in \mathcal{U}_P$, $t$ is in the set $T_Q$.

Since $|P|=n-3$ and there are $n$ indices total, $H_P(t)$ can exist for at most $n-3$ characters $t\in P$ and at most 3 characters $t\not\in P$. We say that the pair $(P,t)$ is \emph{irregular} if $H_P(t)$ exists and $t\not\in P$. The quantity that we will count is the total number of irregular pairs $(P,t)$ over all size $n-3$ sets $P\subset [k']$ and all $t\in [k']$. 

On one hand, as previously mentioned, each set $P$ can only be in at most 3 irregular pairs. Then since there are ${k'\choose n-3}$ sets $P\subset [k']$ of size $n-3$, the total number of irregular pairs is at most $3{k'\choose n-3}$.

On the other hand, Lemma~\ref{lem:config2} implies a lower bound on the number of irregular pairs. By Lemma~\ref{lem:config2}, for \emph{every} size $n-2$ set $Q\subset [k']$, the set $T_Q\subset Q$ is of size $n-4$. Fix sets $Q$ and $T_Q$. We claim that for each $t\in T_Q$, the pair $(Q\setminus\{t\},t)$ is an irregular pair. Firstly, it is clear that $t\not\in Q\setminus\{t\}$. Secondly, $H_{Q\setminus\{t\}}(t)$ exists because $Q\in \mathcal{U}_{Q\setminus\{q\}}$ and $t\in T_Q$. Thus, for each $t\in T_Q$, the pair $(Q\setminus\{t\},t)$ is an irregular pair.
 
 Thus, every size $n-2$ set $Q\subset [k']$ \emph{produces} $n-4$ irregular pairs $(Q\setminus\{t\},t)$. Furthermore, given an irregular pair $(P,t)$, there is only one set that could produce it, namely $P\cup\{t\}$. Then since there are ${k'\choose n-2}$ sets $Q\subset [k']$ of size $n-2$, we have that the total number of irregular pairs is at least $(n-4){k'\choose n-2}$.
 
 Thus, we have shown that the total number of irregular pairs is at most $3{k'\choose n-3}$ and at least $(n-4){k'\choose n-2}$. Therefore, we have reached a contradiction if $3{k'\choose n-3}<(n-4){k'\choose n-2}$ which is true if $n\geq 5$ and $k'>\frac{n^2-4n+6}{n-4}$. In particular, $n\geq 5$, $k'\geq n+7$ satisfy these bounds.
\end{proof}

\section{The remaining parameter regime}\label{sec:3}

%In this section we will prove Theorem~\ref{thm:3}.

\begin{theorem}[restatement of Theorem~\ref{thm:3}]\label{thm:3rr}
For $n\geq 3$, $k\geq 5$, every set of functions $f_1^{n,k},\dots, f_n^{n,k}$ has maximum switching cost at least 3.
\end{theorem}

%Recall from section that having large $n$ presents a new set of challenges. Our proof method in this section is completely different from section . In this section, for the sake of clarity, we will write proofs in the language of the original problem statement, i.e. with agents switching between tasks rather than permutations of multisets. 

\begin{remark*}
We note that the proof framework from Section~\ref{sec:n<k} immediately breaks down if we try to apply it to Theorem~\ref{thm:3rr} for all $n,k$. For example, when $n>k$, there are no size $n$ subsets of $[k]$ so we must instead consider size $n$ \emph{multisets} of $[k]$. Even if we have the same setting of parameters as Theorem~\ref{thm:3r} but we are considering multisets, in step 1 of the proof framework Lemma~\ref{lem:fix} is no longer true. That is, it is not true that for all size 2 multisets $\{x,y\}$ of $[k]$, we have that $\{x,y\}$ $i$-freezes either $x$ or $y$ for some $i$. In particular, suppose $\{x,y\}=\{a,a\}$. Then if is possible that $\pi(\{a,a,b\})=aab$, $\pi(\{a,a,c\})=aca$, and $\pi(\{a,a,d\})=daa$, in which case $a$ is not frozen to any index.  Since the proof framework from Section~\ref{sec:n<k} no longer applies, we develop entirely new techniques in this section. 
\end{remark*}

For the rest of this section we will use the language of the original problem statement rather than that of the problem reformulation.

\subsection{Preliminaries}

To prove the Theorem~\ref{thm:3rr}, we need to show that Theorem~\ref{thm:warm} extends to larger $k$ and $n$. As noted in Section~\ref{sec:remain}, extending to larger $n$ is challenging, while extending to larger $k$ is trivial, as shown in the following lemma.

%We first claim that extending to higher $k$ is trivial. As noted in the our techniques section, the real challenge is to extend to higher $n$. Thus, we only need to prove it for $n\geq 4$ and $k\geq 5$.

\begin{lemma}\label{lem:k}
Fix $n$ and $k$. If there exists a set of functions $f_1^{n,k},\dots, f_n^{n,k}$ with maximum switching cost $D$, then for all $k'<k$, there exists a set of functions $g_1^{n,k'},\dots, g_n^{n,k'}$ with maximum switching cost $D$.
\end{lemma}

\begin{proof}
For each demand vector $\vec{v}$ with $n$ agents and $k$ tasks such that only the first $k'$ entries of $\vec{v}$ are non-zero, let $\vec{v'}$ be the length $k'$ vector consisting of only the first $k'$ entries of $\vec{v}$. We note that the set of all such vectors $\vec{v'}$ is the set of all demand vectors for $n$ agents and $k'$ tasks. Set each $g_i^{n,k'}(\vec{v'})=f_i^{n,k}(\vec{v})$. Then the switching cost for any adjacent pair $(\vec{v'_1},\vec{v'_2})$ with respect to $g_1^{n,k'},\dots, g_n^{n,k'}$ is equal to the switching cost of the corresponding adjacent pair $(\vec{v_1},\vec{v_2})$ with respect to $f_1^{n,k},\dots, f_n^{n,k}$. Thus, the maximum switching cost of $g_1^{n,k'},\dots, g_n^{n,k'}$ is equal to the maximum switching cost of $f_1^{n,k},\dots, f_n^{n,k}$.
\end{proof}

%Let $n$ be the number of agents and let $k$ be the number of tasks. A \emph{demand vector} $\vec{v}=\{d[1],d[2],\dots,d[k]\}$ is a vector of non-negative integers such that $\sum_i d[i]=n$. For all $i$, $d[i]$ is the \emph{demand} of task $i$. For each agent $a$, let $f_a:\vec{v}\rightarrow i\in [k]$ be a function so that given a demand vector as input, $f_a$ outputs agent $a$'s assigned task. For all $i$, the number of agents assigned to task $i$ must match the demand for that task.

%Recall that a pair of demand vectors are \emph{adjacent} if their $\ell_1$ distance is 2 starting with one of the demand vectors, such that exactly one unit of demand is moved from one task to another. 
%A \emph{move} is an ordered pair of adjacent demand vectors. 
\begin{notation*}
We say that an ordered pair of adjacent demand vectors $(\vec{v_1},\vec{v_2})$ is $(s,t)$-adjacent if starting with $\vec{v_1}$ and moving exactly one unit of demand from task $s$ to task $t$ results in $\vec{v_2}$. We say that an agent $a$ is \emph{$(i,j)$-mobile} with respect to an ordered pair of adjacent demand vectors $(\vec{v_1},\vec{v_2})$ if $f_a^{n,k}(\vec{v_1})=i$, $f_a^{n,k}(\vec{v_2})=j$, and $i\not=j$. 

%The \emph{switching cost} of a move $\vec{v_1},\vec{v_2}$ is defined as the number of agents $a$ for which $f_a(\vec{v_1})\not=f_a(\vec{v_2})$. The \emph{switching cost} of a set of functions $f_1,f_2,\dots f_n$ is defined as the maximum switching cost over all moves.

We note that if $(\vec{v_1},\vec{v_2})$ is $(s,t)$-adjacent and has switching cost 2, then for some task $i$, some agent $a$ must be $(s,i)$-mobile and another agent $b$ must be $(i,t)$-mobile. We say that $i$ is the \emph{intermediate} task with respect to $(\vec{v_1},\vec{v_2})$. 
\end{notation*}

%Figure~\ref{fig:prelim} shows how we will depict moves in all of our proofs. Todo use tables instead.

%\begin{figure}[h]
%  \centering
%    \includegraphics[width=0.25\textwidth]{prelim}
%    \caption{A move and its corresponding mobile agents.}\label{fig:prelim}
%
%\end{figure}

%With respect to a move, a \emph{source} is a task that an agent moves from and a \emph{sink} is a task that an agent moves to. For example, the sources for $(\vec{v_1},\vec{v_2})$, are $s$ and $i$ while the sinks are $i$ and $t$. 

\subsection{Proof overview}

We begin by supposing for contradiction that there exists a set of functions $f_1^{n,k},\dots, f_n^{n,k}$ with maximum switching cost 2, and then we prove a series of structural lemmas about such functions. 

As previously mentioned, the main challenge of proving Lemma~\ref{thm:3rr} is handling large $n$. To illustrate this challenge, we repeat the example from Section~\ref{sec:remain}. This example shows that having large $n$ can allow more pairs of adjacent demand vectors to have switching cost 2, making it more difficult to find a pair with switching cost greater than 2. 

Consider the subset $S_i$ of demand vectors in which a particular task $i$ has an unconstrained amount of demand and each remaining task has demand at most $n/(k-1)$. We claim that there exists a set of functions $f_1^{n,k},\dots, f_{n}^{n,k}$ so that every pair of adjacent demand vectors from $S_i$ has switching cost 2. Divide the agents into $k-1$ groups of $n/(k-1)$ agents each, and associate each task except $i$ to such a group of agents. We define the functions $f_1^{n,k},\dots, f_{n}^{n,k}$ so that given any demand vector in $S_i$, the set of agents assigned to each task except $i$ is simply a subset of the group of agents associated with that task (say, the subset of such agents with smallest ID). This is a valid assignment since the demand of each task except $i$ is at most the size of the group of agents associated with that task. The remaining agents are assigned to task $i$. Then, given a pair $(\vec{v},\vec{v'})$ of adjacent demand vectors in $S_i$, whose demands differ only for tasks $s$ and $t$, their switching cost is 2 because the only agents assigned to different tasks between $\vec{v}$ and $\vec{v'}$ are: one agent from each of the groups associated with tasks $s$ and $t$, respectively.

To overcome the challenge illustrated by the above example, our general method is to identify a task that serves the role of task $i$ and then successively move demand out of task $i$ until task $i$ is empty, and thus can no longer serve its original role. We note that in the above example, the task $i$ serves as the intermediate task for all pairs of adjacent demand vectors from $S_i$. Thus, we will choose $i$ to be an intermediate task.

In particular, we show that there is a demand vector $\vec{v}$ so that we can identify tasks $i$ and $t$ with the following important property: if we start with $\vec{v}$ and  move a unit of demand to task $t$ from \emph{any} other task except $i$, the switching cost is 2 and the intermediate task is $i$. 
%Additionally, starting with the demand vector $\vec{v}$, if we move a unit of demand from task $i$ to task $t$, the switching cost is 1. 

Furthermore, we prove that if we start with demand vector $\vec{v}$ and move a unit of demand from task $i$ to task $t$ resulting in demand vector $\vec{v_1}$, then $t$ and $i$ have the important property from the previous paragraph with respect to $\vec{v_1}$. Applying this argument inductively, we show that no matter how many units of demand we successively move from $i$ to $t$, $i$ and $t$ still satisfy the important property with respect to the current demand vector. 

We move demand from $i$ to $t$ until task $i$ is empty. Then, the final contradiction comes from the fact that if we now move a unit of demand from any non-$i$ task to $t$, then the important property implies that the switching cost is 2 and the intermediate task is $i$; however, $i$ is empty and an empty task cannot serve as an intermediate task.\\

\subsection{Proof of Theorem~\ref{thm:3rr}}\label{app:3}
Theorem~\ref{thm:warm} proves Theorem~\ref{thm:3rr} for the case of $n=3$ and $k=5$. Lemma~\ref{lem:k} implies that Theorem~\ref{thm:3rr} also holds for $n=3$ and any $k\geq 5$. Thus, it remains to prove Theorem~\ref{thm:3rr} for $n\geq 4$ and $k\geq 5$. Suppose by way of contradiction that $n\geq 4$, $k\geq 5$, and $f^{n,k}_1,f^{n,k}_2,\dots,f^{n,k}_n$ is a set of functions with switching cost 2.

%Towards the goal of  ref overview, the following lemma says that starting with any demand vector, each task $t$ is one of two types: either all moves to $t$ have switching cost 1, or almost all moves to $t$ have switching cost 2 with the same intermediate task. Or maybe just say this after. Here, say if we start with a demand vector $\vec{v}$ and move a unit of demand to a task $t$ and it has switching cost 2, then moving any other unit of demand to the same task $t$ has has switching cost 2 and the same intermediate task. 

As motivated in the algorithm overview, our first structural lemma concerns tasks $i$ and $t$ such that if we move a unit of demand to task $t$ from any other task except $i$, the switching cost is 2 and the intermediate task is $i$.

\begin{lemma}\label{lem:dest2}
Let $(\vec{v},\vec{v_1})$ be a pair of $(s_1,t)$-adjacent demand vectors with switching cost 2 and intermediate task $i$. Then, for all $\vec{v_2}$ such that $(\vec{v},\vec{v_2})$ are $(s_2,t)$-adjacent for $s_2\not=i$, the pair $(\vec{v},\vec{v_2})$ has switching cost 2, intermediate task $i$, and the same $(i,t)$-mobile agent as $(\vec{v},\vec{v_1})$. 
\end{lemma}

\begin{proof}
Table~\ref{tab1} depicts the proof.

\begin{table}[h]
\centering
\begin{tabular}{c|c|c|c|c}
&$\bm{s_1}$&$\bm{i}$&$\bm{t}$&$\bm{s_2}$\\
\hline
$\vec{\bm{v}}$&$a$&$b$&&$c$\\
\hline
$\vec{\bm{v_1}}$&&$a$&$b$&$c$\\
\hline
\hline
$\vec{\bm{v_2}}$ \bf(case 1) &$a$&$b$&$c$&\\
\hline
$\vec{\bm{v_2}}$ \bf(case 2) &&$b$&$d$&not $c$
\end{tabular}
\caption{Demand vectors and the corresponding assignment of agents. For example, the  row labeled $\vec{v}$ indicates that for the demand vector $\vec{v}$, agent $a$ is assigned to task $s_2$, agent $b$ is assigned to task $i$, and agent $c$ is assigned to task $s_2$. There could also be other agents in the system that are not shown in the table.}
\label{tab1}
\end{table}

With respect to $(\vec{v},\vec{v_1})$, let $a$ be the $(s_1,i)$-mobile agent and let $b$ be the $(i,t)$-mobile agent. Then $a$ and $b$ behave according to rows $\vec{v}$ and $\vec{v_1}$ of Table~\ref{tab1}. 

Suppose by way of contradiction that $(\vec{v},\vec{v_2})$ is \emph{not} as in the lemma statement. That is, either $(\vec{v},\vec{v_2})$ has switching cost 1 or $(\vec{v},\vec{v_2})$ has switching cost 2 and either a different intermediate task from $(\vec{v},\vec{v_1})$ or a different $(i,t)$-mobile agent.

\subparagraph{Case 1. $(\vec{v},\vec{v_2})$ has switching cost 1.} Let $c$ be the mobile agent with respect to $(\vec{v},\vec{v_2})$. Then, $\vec{v_2}$ is as in row $\vec{v_2}$ (case 1) of Table~\ref{tab1}. Also, since $c$ is not mobile with respect to $(\vec{v},\vec{v_1})$, $c$ is assigned to $s_2$ for both $\vec{v}$ and $\vec{v_1}$ as shown in Table~\ref{tab1}. Comparing rows $\vec{v_1}$ and $\vec{v_2}$ (case 1) of Table~\ref{tab1}, it is clear that $(\vec{v_1},\vec{v_2})$ are adjacent and have switching cost 3, since $a$, $b$, and $c$ each switch tasks. This is a contradiction.

\subparagraph{Case 2. $(\vec{v},\vec{v_2})$ has switching cost 2.} Let $i_2$ be the intermediate task of $(\vec{v},\vec{v_2})$ and let $c$ be the $(s_2,i_2)$-mobile agent for $(\vec{v},\vec{v_2})$. Then, for $\vec{v_2}$, $c$ is not assigned to $s_2$, as shown in Table~\ref{tab1}. Let $d$ be the $(i_2,t)$-mobile agent for $(\vec{v},\vec{v_2})$. We note that it is possible that $d=a$, however $d\not=b$ since $d$ is assigned to $i_2$ for $\vec{v}$ while $b$ is assigned to $i$, and $i\not=i_2$. Table~\ref{tab1} shows the positions of $b$ and $d$ (but not $a$) in $\vec{v_2}$. Comparing rows $\vec{v_1}$ and $\vec{v_2}$ (case 2) of Table~\ref{tab1}, it is clear that $(\vec{v_1},\vec{v_2})$ has switching cost 3, since $b$, $d$, and $c$ each switch tasks. Since $(\vec{v_1},\vec{v_2})$ are adjacent, this is a contradiction.
\end{proof}

We have just shown in Lemma~\ref{lem:dest2} that with respect to any demand vector $\vec{v}$, the set of tasks can be split into two distinct types such that every task is of exactly one type.

\begin{definition}[type 1 task] A task $t$ is of \emph{type 1} with respect to a demand vector $\vec{v}$ if when we start with $\vec{v}$ and move a unit of demand from any task to task $t$, the switching cost is 1. 
\end{definition}

\begin{definition}[type 2 task] A task $t$ is of \emph{type 2} with respect to a demand vector $\vec{v}$ if there exists a task $i$ and an agent $a$ such that when we start with $\vec{v}$ and move a unit of demand from any task except $i$ to task $t$, the switching cost is 2, the intermediate task is $i$, and the $(i,t)$-mobile agent is $a$. We say that $a$ is the \emph{intermediate agent} of $t$ with respect to $\vec{v}$.
\end{definition}

\begin{remark*}We note that if $\vec{v}$ only has two non-empty tasks besides $t$, then it is possible that the identity of task $i$ is ambiguous. However, every time we reference a type 2 task we always have the condition that there are at least three non-empty tasks besides $t$ so there will be no ambiguity.
\end{remark*}

%We note that the exception of one task in the definition of type 2 tasks is for the intermediate task $i_t$. That is, if starting from $\vec{v}$ a unit of demand is moved from $i_t$ to $t$, we have not specified any constraints on the switching cost. 

%If we have a type 2 task $t$ with respect to a demand vector $\vec{v}$, Lemma~\ref{lem:ind} addresses the case of moving a unit of demand from the intermediate task $i$ to task $t$. 
As mentioned in the proof overview we wish to successively move demand out of an intermediate task until it is empty. The bulk of the remainder of the proof is to prove the following lemma (Lemma~\ref{lem:ind}), which roughly says that if $t$ is a type 2 task and $i$ is $t$'s intermediate task, then after we move a unit of demand from task $i$ to task $t$, task $t$ remains a type 2 task with intermediate task $i$. Then, by iterating Lemma~\ref{lem:ind}, we show that after moving any amount of demand from task $i$ to task $t$, task $t$ still remains a type 2 task with intermediate task $i$.

%Combining Lemmas~\ref{lem:ind} and \ref{lem:sw1} allows us to successively move demand from $i$ to $t$ with switching cost 1.

\begin{lemma}\label{lem:ind}
Let $\vec{v}$ be a demand vector with at least four non-zero entries. Then there exists a task $t$ such that $t$ is of type 2 with respect to $\vec{v}$ and $\vec{v}$ has at least four non-empty tasks distinct from $t$. Let $i$ be the intermediate task of $t$ with respect to $\vec{v}$. Let $\vec{v'}$ be such that $(\vec{v},\vec{v'})$ are $(i,t)$-adjacent. Then, $t$ is a type 2 task with intermediate task $i$ with respect to $\vec{v'}$. 
\end{lemma}

%Ok, here's what to do: Choose t to be such that there are 4 non-zero entries besides t. We can do this because we know that i is of type 1 by lem and the rest are of type 2 by that other lem, so all empty ones are of type 2 so there exists an empty one with 4 other non-zeros. 
%Todo: add above to pf of lem and copy paste statement to make consistent. And does that mean that like half of the pf is unecessary?
%
%TODO: there's a part in the pf of the final lemma where the commented out lemma is used but it's not necessary and we can rewrite it. It says "let b be the mobile agent for ()" - need to check for this kind of implicit use everywhere! If it's only used in this last lemma then can easily change pf to say first fill in rows v and v_1 and then just try to fit v' in there and it's adjacent to both of them so must share at least one of a or b with each of them and the only way to do this is how it's shown in the table. NO. Actually the commented out pf IS necessarily for the final pf. Above isn't ok bc a and b could be switched in v' row. So, comment back in pf, but I guess include it within the pf of thm 4.1 bc it doesn't need to come before then. 

\paragraph*{Lemma~\ref{lem:ind} implies Theorem~\ref{thm:3rr}.}
Let $\vec{v}$, $t$, and $i$ be as in Lemma~\ref{lem:ind}. We claim that if task $i$ is non-empty in $\vec{v'}$ then the triple ($\vec{v'}$, $t$, $i$) also satisfies the precondition of Lemma~\ref{lem:ind}. This is because if task $i$ is non-empty in $\vec{v'}$ then the set of non-empty tasks in $\vec{v'}$ is a superset of the set of non-empty tasks in $\vec{v'}$. Then since $\vec{v}$ has at least four non-empty tasks distinct from $t$, $\vec{v'}$ also has at least four non-empty tasks distinct from $t$. Also, by Lemma~\ref{lem:ind}, $t$ is a type 2 task with intermediate task $i$ with respect to $\vec{v'}$. Thus, we have shown that if task $i$ is non-empty for $\vec{v'}$ then ($\vec{v'}$, $t$, $i$) satisfy the precondition of Lemma~\ref{lem:ind}. Thus, we can iterate Lemma~\ref{lem:ind}: if we start with $\vec{v}$ and successively move demand from task $i$ to task $t$ until task $i$ is empty, the resulting demand vector $\vec{v''}$ is such that $t$ is a type 2 task with intermediate task $i$. However, it is impossible for $i$ to be an intermediate task with respect to $\vec{v''}$ since $i$ is empty. It remains to prove Lemma~\ref{lem:ind}.

%We will start with an initial demand vector $\vec{v}$ and identify a type 2 task $t$ with intermediate task $i$. Then we will move a unit of demand from $i$ to $t$ with switching cost 1 by Lemma~\ref{lem:sw1}. Then, Lemma~\ref{lem:ind} implies that $t$ is still a type 2 task with intermediate task $i$. Thus, if we again move a unit of demand from $i$ to $t$ the switching cost is 1. We continue moving demand from $i$ to $t$ in this way until $i$ becomes empty. Then we reach a contradiction: Lemma~\ref{lem:ind} says that $t$ is still a type 2 task with intermediate task $i$, however $i$ is empty so it cannot act as an intermediate task.

\subsubsection{Proof of Lemma~\ref{lem:ind}}

The following lemma shows that the pair $(\vec{v},\vec{v'})$ from the statement of Lemma~\ref{lem:ind} has switching cost 1.

\begin{lemma}\label{lem:sw1}
Let $t$ be a type 2 task with intermediate task $i$ with respect to a demand vector $\vec{v}$. Suppose $\vec{v}$ has at least one unit of demand in each of two tasks $s_1$ and $s_2$, both distinct from $t$ and $i$. Let $\vec{v'}$ be the demand vector such that $(\vec{v},\vec{v'})$ is $(i,t)$-adjacent. Then $(\vec{v},\vec{v'})$ has switching cost 1.
%
%Let $(\vec{v},\vec{v_1})$ be a pair of $(s_1,t)$-adjacent demand vectors with switching cost 2 and let $(\vec{v},\vec{v_2})$ be a pair of $(s_2,t)$-adjacent demand vectors with switching cost 2 and $s_1\not=s_2$. Let $i$ be the intermediate task of both $(\vec{v},\vec{v_1})$ and $(\vec{v},\vec{v_2})$. 
\end{lemma}

\begin{proof}
 Table~\ref{tab2} depicts the proof.

\begin{table}[h]
\centering
\begin{tabular}{c|c|c|c|c}
&$\bm{s_1}$&$\bm{i}$&$\bm{t}$&$\bm{s_2}$\\
\hline
$\vec{\bm{v}}$&$a$&$b$&not $d$&$c$\\
\hline
$\vec{\bm{v_1}}$&&$a$&$b$&$c$\\
\hline
$\vec{\bm{v_2}}$&$a$&$c$&$b$, not $d$&\\
\hline
\hline
$\vec{\bm{v'}}$ \bf(case 1) &&&$a$, not $b$&$c$\\
\hline
$\vec{\bm{v'}}$ \bf(case 2) &$a$&&$d$, not $b$&$c$
\end{tabular}
\caption{Demand vectors and the corresponding assignment of agents.}
\label{tab2}
\end{table}
Let $\vec{v_1}$ be such that $(\vec{v},\vec{v_1})$ is $(s_1,t)$-adjacent and let $\vec{v_2}$ be such that $(\vec{v},\vec{v_2})$ is $(s_2,t)$-adjacent.
With respect to $(\vec{v},\vec{v_1})$, let $a$ be the $(s_1,i)$-mobile agent and let $b$ be the $(i,t)$-mobile agent. Then $a$ and $b$ behave according to rows $\vec{v}$ and $\vec{v_1}$ of Table~\ref{tab1}.

With respect to $(\vec{v},\vec{v_2})$, let $c$ be the $(s_2,i)$-mobile agent. From Lemma~\ref{lem:dest2} we know that $b$ is the $(i,t)$-mobile agent for $(\vec{v},\vec{v_2})$. Thus $b$ and $c$ behave according to rows $\vec{v}$ and $\vec{v_2}$ of Table~\ref{tab1}.

Suppose by way of contradiction that $(\vec{v},\vec{v'})$ has switching cost 2. Let $i'\not=i,t$ be the intermediate task.  We condition on the $(i',t)$-mobile agent. We already know that it is not $b$ since $b$ is assigned to $i$ for $\vec{v}$.

\subparagraph{Case 1. the $(i',t)$-mobile agent for $(\vec{v},\vec{v'})$ is $a$ or $c$.} Suppose the $(i',t)$-mobile agent for $(\vec{v},\vec{v'})$ is $a$, as shown in row $\vec{v'}$ (case 1) of Table~\ref{tab2}. If the $(i',t)$-mobile agent is $c$, the argument is identical. Since we are assuming $c$ is not a mobile-agent, $c$ is assigned to $s_2$ in $\vec{v'}$ as shown in Table~\ref{tab2}. Also, since $a$ is the only $(i',t)$-mobile agent, we know that $b$ is not assigned to $t$ in $\vec{v'}$ as shown in Table~\ref{tab2}. Comparing rows $\vec{v_2}$ and $\vec{v'}$ (case 1) of Table~\ref{tab2}, it is clear that $(\vec{v_2},\vec{v'})$ have switching cost 3, since $a$, $b$, and $c$ all switch tasks. Also, $\vec{v_2}$ and $\vec{v'}$ are adjacent since both are the result of starting with $\vec{v}$ and moving one unit of demand from some task to task $t$. This is a contradiction.

\subparagraph{Case 2. the $(i',t)$-mobile agent for $(\vec{v},\vec{v'})$ is neither $a$ nor $c$.} Let $d$ be the $(i',t)$-mobile agent for $(\vec{v},\vec{v_3})$. Then $d$, $a$, and $c$ are assigned as in row $\vec{v'}$ (case 2) of Table~\ref{tab2}. Also, since $d$ is the only $(i',t)$-mobile agent and $d\not=b$, we know that $b$ is not assigned to $t$ in $\vec{v'}$, as shown in Table~\ref{tab2}. 

Also, since $d$ is the $(i',t)$-mobile agent for $(\vec{v},\vec{v'})$, we know that $d$ is not assigned to $t$ in $\vec{v}$, as shown in Table~\ref{tab2}. Then, since $b$ is the only $(i,t)$-mobile agent for $(\vec{v},\vec{v_2})$, we know that $d$ is also not assigned to $t$ for $\vec{v_2}$, as shown in Table~\ref{tab2}.

Comparing rows $\vec{v_2}$ and $\vec{v'}$ (case 2) of Table~\ref{tab2}, it is clear that $(\vec{v_2},\vec{v_3})$ have switching cost 3, since $d$, $b$, and $c$ all switch tasks. This is a contradiction since $\vec{v_2}$ and $\vec{v'}$ are adjacent.
\end{proof}

Next, we prove another structural lemma concerning the intermediate task $i$, which says that task $i$ is of type 1.

\begin{lemma}\label{lem:tech}
Let $\vec{v}$ be a demand vector with at least four non-zero entries and let $t$ be a type 2 task with intermediate task $i$. Then, task $i$ is of type 1 with respect to $\vec{v}$.
\end{lemma}

\begin{proof}
Table~\ref{tabb} depicts the proof. 

\begin{table}[h]
\centering
\begin{tabular}{c|c|c|c}
&$\bm{s}$&$\bm{i}$&$\bm{t}$\\
\hline
$\vec{\bm{v}}$&$a$&$b$, not $c$&\\
\hline
$\vec{\bm{v_1}}$&&$a$, not $c$&$b$\\
\hline
$\vec{\bm{v_2}}$&&$b$, $c$, not $a$&\\
\end{tabular}
\caption{Demand vectors and the corresponding assignment of agents.}
\label{tabb}
\end{table} 

Suppose for contradiction that task $i$ is of type 2 with respect to $\vec{v}$, and let $i'$ be the intermediate task. Since $\vec{v}$ has at least four non-zero entries, there exists a task $s\not=i'$ that is non-empty for $\vec{v}$. Let $\vec{v_1}$ be such that $(\vec{v},\vec{v_1})$ are $(s,t)$-adjacent. For  $(\vec{v},\vec{v_1})$, let $a$ be the $(s,i)$-mobile agent and let $b$ be the $(i,t)$-mobile agent, as shown in Table~\ref{tabb}.

Letting $\vec{v_2}$ be such that $(\vec{v},\vec{v_2})$ are $(s,i)$-adjacent, $(\vec{v},\vec{v_2})$ has switching cost 2 since $s\not=i'$. Thus, $a$ does \emph{not} switch  to task $i$ for $(\vec{v},\vec{v_2})$, as shown in Table~\ref{tabb}. Also, $b$ remains in task $i$ for $(\vec{v},\vec{v_2})$, as shown in Table~\ref{tabb}. Let $c$ be the mobile agent for $(\vec{v},\vec{v_2})$ that switches to task $i$. Then agent $c$ is not assigned to task $i$ with respect to $\vec{v}$ or $\vec{v_1}$ and is assigned to task $i$ with respect to $\vec{v_2}$, as shown in Table~\ref{tabb}. 

We note that $(\vec{v_1},\vec{v_2})$ are $(t,i)$-adjacent while they differ on the assignment of agents $a$, $b$, and $c$, a contradiction.
\end{proof}

Later, we will prove the following lemma (Lemma~\ref{lem:all2int}), which says that (under certain conditions) there is \emph{at most one} task of type 1. Combining this with Lemma~\ref{lem:tech} allows us to say that every task except for intermediate task $i$ is of type 2, which will be a useful structural property. 

\begin{lemma}\label{lem:all2int}
For any demand vector $\vec{v}$ with at least four non-zero entries, there is at most one task of type 1.
\end{lemma}

In order to prove Lemma~\ref{lem:all2int}, we will prove two structural lemmas. The following simple lemma is useful (but may at first appear unrelated).

\begin{lemma}\label{lem:ss}
Let $(\vec{v},\vec{v_1})$ be a pair of $(s,t_1)$-adjacent demand vectors with switching cost 2 and intermediate task $i_1$. Let $(\vec{v},\vec{v_2})$ be a pair of $(s,t_2)$-adjacent demand vectors with switching cost 2 and intermediate task $i_2$. Then it is \emph{not} the case that $s$, $i_1$, $t_1$, $i_2$, and $t_2$ are all distinct.
\end{lemma}

\begin{proof}
Suppose by way of contradiction that $s$, $i_1$, $t_1$, $i_2$, and $t_2$ are all distinct. Table~\ref{tab3} depicts the proof.

\begin{table}[h]
\centering
\begin{tabular}{c|c|c|c|c|c}
&$\bm{s}$&$\bm{i_1}$&$\bm{t_1}$&$\bm{i_2}$&$\bm{t_2}$\\
\hline
$\vec{\bm{v}}$&$a$&$b$&&$c$&\\
\hline
$\vec{\bm{v_1}}$&&$a$&$b$&$c$&\\
\hline
$\vec{\bm{v_2}}$&&$b$, not $a$&&&$c$\\
\end{tabular}
\caption{Demand vectors and the corresponding assignment of agents.}
\label{tab3}
\end{table}
With respect to $(\vec{v},\vec{v_1})$, let $a$ be the $(s,i_1)$-mobile agent and let $b$ be the $(i_1,t_1)$-mobile agent. Then $a$ and $b$ behave according to rows $\vec{v}$ and $\vec{v_1}$ of Table~\ref{tab3}.

With respect to $(\vec{v},\vec{v_2})$, let $c$ be the $(i_2,t_2)$-mobile agent. Then $c$ behaves according to Table~\ref{tab3}. Since $i_1\not=s, i_2,t_2$, we know that $b$ is in the same position in $\vec{v}$ and $\vec{v_2}$. For the same reason, $a$ does not move to $i_1$ with respect to $(\vec{v},\vec{v_2})$, as shown in Table~\ref{tab3}. 

Comparing rows $\vec{v_1}$ and $\vec{v_2}$ in Table~\ref{tab3}, it is clear that $(\vec{v_1},\vec{v_2})$ have switching cost 3. Since $\vec{v_1}$ and $\vec{v_2}$ are adjacent, this is a contradiction.
\end{proof}

The following lemma is a weaker version of Lemma~\ref{lem:all2int}, which says that there is at least one task of type 2.

\begin{lemma}\label{lem:one2}
For any demand vector $\vec{v}$ with at least two non-zero entries, there is at least one task of type 2.
\end{lemma}

\begin{proof}
Table~\ref{tab4} depicts the proof. 

\begin{table}[h]
\centering
\begin{tabular}{c|c|c|c|c}
&$\bm{s_1}$&$\bm{s_2}$&$\bm{s_4}$&$\bm{s_5}$\\
\hline
$\vec{\bm{v}}$&$a$&$b$&&\\
\hline
$\vec{\bm{v_1}}$&&$b$&$a$&\\
\hline
$\vec{\bm{v_2}}$&$a$&&&$b$\\
\hline
$\vec{\bm{v_3}}$&$a$&&$b$&\\
\hline
$\vec{\bm{v_4}}$&&$b$&&$a$\\
\hline
\hline
$\vec{\bm{v_5}}$ (case 1)&&&$a$&$b$\\
\hline
$\vec{\bm{v_5}}$ (case 2)&&&$b$&$a$\\

\end{tabular}
\caption{Demand vectors and the corresponding assignment of agents.}
\label{tab4}
\end{table} 

Suppose by way of contradiction that every task is of type 1 with respect to $\vec{v}$. That is, for all $\vec{v'}$ adjacent to $\vec{v}$, the pair $(\vec{v},\vec{v'})$ has switching cost 1. Let $s_1$, $s_2$, and $s_3$ be non-empty tasks with respect to $\vec{v}$ ($s_3$ is not shown in Table~\ref{tab4}). Let $s_4$ and $s_5$ be additional tasks. Let $\vec{v_1}$ be such that $(\vec{v},\vec{v_1})$ are $(s_1,s_4)$-adjacent with mobile agent $a$. Let $\vec{v_2}$ be such that $(\vec{v},\vec{v_2})$ are $(s_2,s_5)$-adjacent with mobile agent $b$. The assignments of agents $a$, $b$, and $c$ for vectors $\vec{v}$, $\vec{v_1}$, and $\vec{v_2}$ are shown in Table~\ref{tab4}.

Let $\vec{v_3}$ be such that $(\vec{v},\vec{v_3})$ are $(s_2,s_4)$-adjacent. We know that $(\vec{v},\vec{v_3})$ has switching cost 1 but we do not know whether the mobile agent is $b$ or some other agent.  Similarly, let $\vec{v_4}$ be such that $(\vec{v},\vec{v_4})$ are $(s_1,s_5)$-adjacent. We know that $(\vec{v},\vec{v_4})$ has switching cost 1 but we do not know whether the mobile agent is $a$ or some other agent. We claim that the mobile agent for  $(\vec{v},\vec{v_4})$ is $a$, and symmetrically the mobile agent for $(\vec{v},\vec{v_3})$ is $b$, as shown in Table~\ref{tab4}.  

Suppose for contradiction that the mobile agent for $(\vec{v},\vec{v_4})$ is some agent $d\not=a$. 
%Still, the assignments of $b$ and $c$ in $\vec{v_4}$ are as in Table~\ref{tab4}. 
Let $\vec{v_5}$ be such that $(\vec{v_1},\vec{v_5})$ are $(s_2,s_5)$-adjacent. Note that $\vec{v_5}$ is adjacent to both $\vec{v_4}$ and $\vec{v_2}$. $\vec{v_2}$ places agent $b$ (and not $d$) in task $s_5$, and $\vec{v_4}$ places agent $d$ (and not $b$) in task $s_5$. Then, since $\vec{v_5}$ can only place at most one of $b$ or $d$ in task $s_5$, $\vec{v_5}$ must disagree with either $\vec{v_2}$ or $\vec{v_4}$ on the assignment of \emph{both} $b$ and $d$. Also, $\vec{v_1}$ places agent $a$ in task $s_4$ while both $\vec{v_2}$ and $\vec{v_4}$ place agent $a$ in task $s_1$. Since $(\vec{v_1},\vec{v_5})$ are $(s_2,s_5)$-adjacent and have switching cost at most 2, agent $a$ cannot switch from task $s_4$ to task $s_1$ when the demand vector changes from $\vec{v_1}$ to $\vec{v_5}$. Thus, $\vec{v_5}$ disagrees with both $\vec{v_2}$ and $\vec{v_4}$ on the assignment of agent $a$. Thus, we have shown that $\vec{v_5}$ disagrees with either $\vec{v_2}$ or $\vec{v_4}$ on the assignment of $a$, $b$, and $d$, a contradiction. Therefore, the mobile agent for $(\vec{v},\vec{v_4})$ is $a$, as shown in Table~\ref{tab4}. By the same argument, the mobile agent for  $(\vec{v},\vec{v_3})$ is $b$, as shown in Table~\ref{tab4}.

Now, consider $(\vec{v_1},\vec{v_5})$, which are $(s_2,s_5)$-adjacent. We first claim that no agents besides $a$ and $b$ are mobile for $(\vec{v_1},\vec{v_5})$. Like the previous paragraph, Note that $\vec{v_5}$ is adjacent to both $\vec{v_4}$ and $\vec{v_2}$. Again, $\vec{v_2}$ places agent $b$ (and not $a$) in task $s_5$ and $\vec{v_4}$ places agent $a$ (and not $b$) in task $s_5$. Then, since $\vec{v_5}$ can only place at most one of $a$ or $b$ in task $s_5$, $\vec{v_5}$ must disagree with either $\vec{v_2}$ or $\vec{v_4}$ on the assignment of \emph{both} $a$ and $b$. Then since $\vec{v_1}$, $\vec{v_2}$, and $\vec{v_4}$ all agree on the assignment of every agent except $a$ and $b$, no other agent $c$ can be mobile for $(\vec{v_1},\vec{v_5})$, because then $\vec{v_5}$ would disagree with either $\vec{v_2}$ or $\vec{v_4}$ on the placement of agents $a$, $b$, and $c$. Thus, no agents besides $a$ and $b$ are mobile for $(\vec{v_1},\vec{v_5})$.

Therefore, we have two cases:
\subparagraph{Case 1: $(\vec{v_1},\vec{v_5})$ has switching cost 1.} In this case the only mobile agent for $(\vec{v_1},\vec{v_5})$ is $b$ as indicated in row $\vec{v_5}$ (case 1) of Table~\ref{tab4}.

Table~\ref{tab4} shows that $(\vec{v_5},\vec{v_4})$ has switching cost 2 where $a$ is $(s_4,s_5)$-mobile and $b$ is $(s_5,s_2)$-mobile. Thus, with respect to $\vec{v_5}$, task $s_2$ is of type 2 with intermediate task $s_5$. 

Similarly, Table~\ref{tab4} shows that $(\vec{v_5},\vec{v_3})$ has switching cost 2 where $b$ is $(s_5,s_4)$-mobile and $a$ is $(s_4,s_1)$-mobile. Thus, with respect to $\vec{v_5}$, task $s_1$ is of type 2 with intermediate task $s_4$. 

We observe that the combination of the previous two paragraphs violates Lemma~\ref{lem:ss}. Recall that $s_3$ is a task that is non-empty for $\vec{v}$, and thus also $\vec{v_5}$. We apply Lemma~\ref{lem:ss} with parameters $(\vec{v}, s, t_1, t_2)=(\vec{v_5},s_3,s_2,s_1)$. Then, from the previous two paragraphs, the intermediate task for task $s_2$ is task $s_5$ and the intermediate task for task $s_1$ is task $s_4$, so $i_1$ and $i_2$ are tasks $s_5$ and $s_4$, respectively. Then $s$, $t_1$, $t_2$, $i_1$, and $i_2$ are all distinct, which violates Lemma~\ref{lem:ss}.

\subparagraph{Case 2: $(\vec{v_1},\vec{v_5})$ has switching cost 2.} In this case for $(\vec{v_1},\vec{v_5})$, $b$ is $(s_2,s_4)$-mobile and $a$ is $(s_4,s_5)$-mobile, as indicated in row $\vec{v_5}$ (case 2) of Table~\ref{tab4}.

Taking the reverse, $(\vec{v_5},\vec{v_1})$ has switching cost 2 where $a$ is $(s_5,s_4)$-mobile and $b$ is $(s_4,s_2)$-mobile. Thus, with respect to $\vec{v_5}$, task $s_2$ is of type 2 with intermediate task $s_4$.

Similarly, Table~\ref{tab4} shows that $(\vec{v_5},\vec{v_2})$ has switching cost 2 where $b$ is $(s_4,s_5)$-mobile and $a$ is $(s_5,s_1)$-mobile. Thus, with respect to $\vec{v_5}$, task $s_1$ is of type 2 with intermediate task $s_5$. 

We observe that the combination of the previous two paragraphs violates Lemma~\ref{lem:ss}. We apply Lemma~\ref{lem:ss} with parameters $(\vec{v}, s, t_1, t_2)=(\vec{v_5},s_3,s_2,s_1)$. Then, from the previous two paragraphs, the intermediate task for task $s_2$ is task $s_4$ and the intermediate task for task $s_1$ is task $s_5$, so $i_1$ and $i_2$ are tasks $s_4$ and $s_5$ respectively. Then $s$, $t_1$, $t_2$, $i_1$, and $i_2$ are all distinct, which violates Lemma~\ref{lem:ss}.
\end{proof}

We are now ready to prove Lemma~\ref{lem:all2int}.

\begin{lemma}[Restatement of Lemma~\ref{lem:all2int}]\label{lem:all2}
For any demand vector $\vec{v}$ with at least four non-zero entries, there is at most one task of type 1.
\end{lemma}

\begin{proof}
 Suppose by way of contradiction that there are two type 1 tasks $t_1$, $t_1'$ with respect to $\vec{v}$. By Lemma~\ref{lem:one2}, there must be at least one type 2 task with respect to $\vec{v}$. Let $t_2$ be a type 2 task with respect to $\vec{v}$, choosing an empty task if possible. Let $i$ be the intermediate task for $t_2$ with respect to $\vec{v}$. Either $i\not=t_1$ or $i\not=t_1'$. Without loss of generality, suppose $i\not=t_1$. Let $s_1$ be a non-empty task with $s_1\not=i,t,t_2$. 

For the construction, we require an additional non-empty task $s_2$ with $s_2\not=s_1,i,t,t_2$. However, such a task might not exist since we are only assuming that $\vec{v}$ has at least four non-zero entries. We note that we cannot assume that $\vec{v}$ has at least five non-zero entries because Theorem~\ref{thm:3rr} only assumes that the total number of agents is at least four. Thus, the existence of $s_2$ is a technicality that is only important when the number of agents is exactly four.  We will first assume that such a task $s_2$ exists and later we will show that $s_2$ must indeed exist. Table~\ref{tab5} depicts the proof. 

\begin{table}[h]
\centering
\begin{tabular}{c|c|c|c|c|c}
&$\bm{s_1}$&$\bm{i}$&$\bm{t_1}$&$\bm{t_2}$&$\bm{s_2}$\\
\hline
$\vec{\bm{v}}$&$a$&$b$&&&$c$\\
\hline
$\vec{\bm{v_1}}$&&$a$&&$b$&$c$\\
\hline
$\vec{\bm{v_2}}$&&$b$&$a$&&$c$\\
\hline
$\vec{\bm{v_3}}$&&$c$&$a$&$b$&\\
\hline
$\vec{\bm{v_4}}$&$a$&$b$, not $c$&&&\\
\end{tabular}
\caption{Demand vectors and the corresponding assignment of agents.}
\label{tab5}
\end{table}
Let $\vec{v_1}$ be such that $(\vec{v},\vec{v_1})$ are $(s_1,t_2)$-adjacent. Since $t_2$ is of type 2 with intermediate task $i$ with respect to $\vec{v}$, $(\vec{v},\vec{v_1})$ has switching cost 2 and intermediate task $i$. Let $a$ be the $(s,i)$-mobile agent and let $b$ be the $(i,t_2)$-mobile agent, as shown in Table~\ref{tab5}.  

Let $\vec{v_2}$ be such that $(\vec{v},\vec{v_2})$ are $(s_1,t_1)$-adjacent. Since $t_1$ is of type 1 with respect to $\vec{v}$, $(\vec{v},\vec{v_2})$ has switching cost 1. Thus, agents $b$ and $c$ are not mobile for $(\vec{v},\vec{v_3})$ as shown in Table~\ref{tab5}. If $a$ is the mobile agent for $(\vec{v},\vec{v_2})$ then $a$ is assigned to task $t_1$ for $\vec{v_2}$ and otherwise $a$ is assigned to task $s$. In either case, the assignment of both $a$ and $b$ differs between $\vec{v_2}$ and $\vec{v_1}$. Since $(\vec{v_1},\vec{v_2})$ are adjacent, they cannot differ on the assignment of any other agents besides $a$ and $b$. Thus, $(\vec{v_1},\vec{v_2})$ has $b$ as its $(t_2,i)$-mobile agent and $a$ as its $(i,t_1)$ mobile agent, as shown in Table~\ref{tab5}. Thus, with respect to $\vec{v_1}$, task $t_1$ is of type 2 with intermediate task $i$ and intermediate agent $a$. We have also shown that the mobile agent for $(\vec{v},\vec{v_2})$ is $a$.

Let $\vec{v_3}$ be such that $(\vec{v_1},\vec{v_3})$ are $(s_2,t_1)$-adjacent. Since $t_1$ is of type 2 with intermediate task $i$ and intermediate agent $a$, $(\vec{v_1},\vec{v_3})$ has switching cost 2 and $a$ is the $(i,t_1)$-mobile agent. Let $c$ be the $(s_2,i)$-mobile agent for $(\vec{v_1},\vec{v_3})$. Then, $\vec{v_3}$ is as is in Table~\ref{tab5}.

Let $\vec{v_4}$ be such that $(\vec{v},\vec{v_4})$ are $(s_2,t_1)$-adjacent. Since $t_1$ is of type 1 with respect to $\vec{v}$, $(\vec{v},\vec{v_4})$ has switching cost 1. Thus, neither $a$ nor $b$ are mobile agents for $(\vec{v},\vec{v_4})$, as shown in Table~\ref{tab5}. If $c$ is the mobile agent for $(\vec{v},\vec{v_4})$ then $\vec{v_4}$ assigns $c$ to task $t_1$ and otherwise $\vec{v_4}$ assigns $c$ to $s_2$. In either case, $\vec{v_4}$ does not assign $c$ to task $i$, as shown in Table~\ref{tab5}. We note that $(\vec{v_3},\vec{v_4})$ are $(t_2,s_1)$-adjacent, however according to Table~\ref{tab5}, $\vec{v_3}$ and $\vec{v_4}$ disagree on the assignment of $a$, $b$, and $c$, a contradiction.

In the above construction, we assumed the existence of task $s_2$. It remains to show that there indeed exists a non-empty task $s_2$ with $s_2\not=s,i,t,t_2$. As mentioned previously, the existence of $s_2$ is a technicality that is only important when the number of agents is exactly four. Thus, the remainder of the proof merely addresses a technicality. Table~\ref{tab6} depicts the remainder of the proof. 

Suppose by way of contradiction that $s_1$, $i$, $t$, and $t_2$ are the only non-empty tasks in $\vec{v}$. Let $s'$ be an empty task. The task $t_2$ was initially chosen to be an \emph{empty} type 2 task if one exists. Since $t_2$ is non-empty, we know that $s'$ is of type 1. 

Let $\vec{v_1}$ be such that $(\vec{v},\vec{v_1})$ are $(s_1,t_2)$-adjacent. Since $t_2$ is of type 2 with intermediate task $i$ with respect to $\vec{v}$, $(\vec{v},\vec{v_1})$ has switching cost 2 and intermediate task $i$. For $(\vec{v},\vec{v_1})$, let $a$ be the $(s_1,i)$-mobile agent and let $b$ be the $(i,t_2)$-mobile agent, as shown in Table~\ref{tab6}. 

\begin{table}[h]
\centering
\begin{tabular}{c|c|c|c|c|c}
&$\bm{s_1}$&$\bm{i}$&$\bm{t_1}$&$\bm{t_2}$&$\bm{s'}$\\
\hline
$\vec{\bm{v}}$&$a$&$b$&$c$&&\\
\hline
$\vec{\bm{v_1}}$&&$a$&$c$&$b$&\\
\hline
$\vec{\bm{v_2}}$&$a$&$b$&&&$c$\\
\hline
$\vec{\bm{v_3}}$&&$b$&$c$&&$a$\\
\hline
$\vec{\bm{v_4}}$&&&&$b$&$a$, not $c$\\
\end{tabular}
\caption{Demand vectors and the corresponding assignment of agents.}
\label{tab6}
\end{table} 

Let $\vec{v_2}$ be such that $(\vec{v},\vec{v_2})$ are $(t,s')$-adjacent. Since $s'$ is of type 1 with respect to $\vec{v}$, $(\vec{v},\vec{v_2})$ has switching cost 1. Let $c$ be the mobile agent for $(\vec{v},\vec{v_2})$ as shown in Table~\ref{tab6}. 

Let $\vec{v_3}$ be such that $(\vec{v},\vec{v_3})$ are $(s_1,s')$-adjacent. Since $s'$ is of type 1 with respect to $\vec{v}$, $(\vec{v},\vec{v_3})$ has switching cost 1. Thus, agents $b$ and $c$ are not mobile for $(\vec{v},\vec{v_3})$ as shown in Table~\ref{tab6}. If $a$ is the mobile agent for $(\vec{v},\vec{v_3})$ then $a$ is assigned to task $s'$ for $\vec{v_3}$ and otherwise $a$ is assigned to task $s$. In either case, the assignment of both $a$ and $b$ differs between $\vec{v_3}$ and $\vec{v_1}$. Since $(\vec{v_1},\vec{v_3})$ are adjacent, they cannot differ on the assignment of any other agents besides $a$ and $b$. Thus, $(\vec{v_1},\vec{v_3})$ has $b$ as its $(t_2,i)$-mobile agent and $a$ as its $(i,s')$ mobile agent, as shown in Table~\ref{tab6}. Thus, with respect to $\vec{v_1}$, task $s'$ is of type 2 with intermediate task $i$ and intermediate agent $a$. We have also shown that the mobile agent for $(\vec{v},\vec{v_3})$ is $a$.

Let $\vec{v_4}$ be such that $(\vec{v_1},\vec{v_4})$ are $(t_1,s')$-adjacent. Since task $s'$ is of type 2 with intermediate task $i$ and intermediate agent $a$ with respect to $\vec{v_1}$, $(\vec{v_1},\vec{v_4})$ have switching cost 2 and intermediate task $i$ and $(i,s')$-mobile agent $a$. Regardless of the $(t_1,i)$-mobile agent for $(\vec{v_1},\vec{v_4})$, agent $c$ is not assigned to task $s'$ for $\vec{v_4}$, as shown in Table~\ref{tab6}.

We note that $(\vec{v_2},\vec{v_4})$ are $(s_1,t_2)$-adjacent, however according to Table~\ref{tab6} they disagree on the assignment of $a$, $b$, and $c$, a contradiction.
\end{proof}

We are now ready to prove Lemma~\ref{lem:ind}.

\begin{lemma}[restatement of Lemma~\ref{lem:ind}]\label{lem:indre}
Let $\vec{v}$ be a demand vector with at least four non-zero entries. Then there exists a task $t$ such that $t$ is of type 2 with respect to $\vec{v}$ and $\vec{v}$ has at least four non-empty tasks distinct from $t$. Let $i$ be the intermediate task of $t$ with respect to $\vec{v}$. Let $\vec{v'}$ be such that $(\vec{v},\vec{v'})$ are $(i,t)$-adjacent. Then, $t$ is a type 2 task with intermediate task $i$ with respect to $\vec{v'}$. 
\end{lemma}

\begin{proof}
First we will show that there exists a task $t$ such that $t$ is of type 2 with respect to $\vec{v}$ and $\vec{v}$ has at least four non-empty tasks distinct from $t$.
By Lemma~\ref{lem:one2}, there exists a task of type 2 with respect to $\vec{v}$. Suppose for contradiction that every such task $t'$ is such that $\vec{v}$ has less than four non-empty tasks distinct from $t'$. Then, every empty task is of type 1, because if there were an empty task $t''$ of type 2 then $\vec{v}$ would have at least four non-empty tasks distinct from $t''$ since $\vec{v}$ has at least four non-zero entries. By Lemma~\ref{lem:tech}, task $i$ is of type 1 with respect to $\vec{v}$. Since $i$ is an intermediate task with respect to $\vec{v}$, $i$ is non-empty. Then, by Lemma~\ref{lem:all2int} $i$ is the only type 1 task with respect to $\vec{v}$. Thus, we have shown that every empty task is of type 1 and there are no empty type 1 tasks, so \emph{every} task is non-empty with respect to $\vec{v}$. Then since there are at least 5 tasks total, there are at least four non-empty tasks distinct from $t'$, a contradiction.

Now, we will show that $t$ is a type 2 task with intermediate task $i$, with respect to $\vec{v'}$. Since $\vec{v}$ has at least four non-empty tasks excluding task $t$,  $\vec{v'}$ has at least four non-zero entries. Thus, we can apply Lemma~\ref{lem:all2} to $\vec{v'}$.

We claim that task $i$ is of type 1 with respect to $\vec{v'}$. The proof of this claim is simple and is depicted in Table~\ref{tab7}. 

\begin{table}[h]
\centering
\begin{tabular}{c|c|c|c}
&$\bm{s}$&$\bm{i}$&$\bm{t}$\\
\hline
$\vec{\bm{v}}$&$a$&$b$\\
\hline
$\vec{\bm{v'}}$&$a$&&$b$\\
\hline
$\vec{\bm{v_1}}$&&$a$&$b$\\
\end{tabular}
\caption{Demand vectors and the corresponding assignment of agents.}
\label{tab7}
\end{table} 
By Lemma~\ref{lem:sw1}, $(\vec{v},\vec{v'})$ have switching cost 1. Let $b$ be the mobile agent for $(\vec{v},\vec{v'})$, as shown in Table~\ref{tab7}. Let $s$ be a task with $s\not=i,t$. Let $\vec{v_1}$ be such that $(\vec{v},\vec{v_1})$ are $(s,t)$-adjacent. Since task $t$ is of type 2 with respect to $\vec{v}$, $(\vec{v},\vec{v_1})$ has switching cost 2 and intermediate task $i$. Let $a$ be the $(s,i)$-mobile agent, as shown in Table~\ref{tab7}. 

If the $(i,t)$-mobile agent for $(\vec{v},\vec{v_1})$ is some agent $c\not=b$, then $\vec{v'}$ and $\vec{v_1}$ disagree on the position of $a$, $b$, and $c$, which is impossible since $(\vec{v'},\vec{v_1})$ are adjacent. Thus, $b$ is the $(i,t)$-mobile agent for $(\vec{v},\vec{v_1})$ as shown in Table~\ref{tab7}. From Table~\ref{tab7}, it is clear that $(\vec{v'},\vec{v_1})$ are $(s,i)$-adjacent and have switching cost 1, and $(\vec{v'},\vec{v})$ are $(t,i)$-adjacent and have switching cost 1. Thus, $i$ is of type 1 with respect to $\vec{v'}$.

By Lemma~\ref{lem:all2}, task $i$ is the only task of type 1 with respect to $\vec{v'}$, so task $t$ is of type 2. It remains to show that task $t$ has intermediate task $i$ with respect to $\vec{v'}$. If task $t$ has a different intermediate task $i'$ with respect to $\vec{v'}$, then by Lemma~\ref{lem:tech}, task $i'$ is of type 1 with respect to $\vec{v'}$, but we already know that task $i$ is the only task of type 1 with respect to $\vec{v'}$. Thus, task $t$ has intermediate task $i$ with respect to $\vec{v'}$.
\end{proof}

\section{Acknowledgments} We would like to thank Yufei Zhao for a discussion.

\bibliography{proj}

\begin{thebibliography}{10}

\bibitem{beshers2001models}
Samuel~N Beshers and Jennifer~H Fewell.
\newblock Models of division of labor in social insects.
\newblock {\em Annual review of entomology}, 46(1):413--440, 2001.

\bibitem{castello2013task}
Eduardo Castello, Tomoyuki Yamamoto, Yutaka Nakamura, and Hiroshi Ishiguro.
\newblock Task allocation for a robotic swarm based on an adaptive response
  threshold model.
\newblock In {\em 2013 13th International Conference on Control, Automation and
  Systems (ICCAS 2013)}, pages 259--266. IEEE, 2013.

\bibitem{chen2015cooperation}
Jianing Chen.
\newblock {\em Cooperation in Swarms of Robots without Communication}.
\newblock PhD thesis, University of Sheffield, 2015.

\bibitem{conlon2010hypergraph}
David Conlon, Jacob Fox, and Benny Sudakov.
\newblock Hypergraph ramsey numbers.
\newblock {\em Journal of the American Mathematical Society}, 23(1):247--266,
  2010.

\bibitem{cornejo2014task}
Alejandro Cornejo, Anna Dornhaus, Nancy Lynch, and Radhika Nagpal.
\newblock Task allocation in ant colonies.
\newblock In {\em International Symposium on Distributed Computing}, pages
  46--60. Springer, 2014.

\bibitem{dornhaus2018self}
Anna Dornhaus, Nancy Lynch, Frederik Mallmann-Trenn, Dominik Pajak, and
  Tsvetomira Radeva.
\newblock Self-stabilizing task allocation in spite of noise.
\newblock {\em arXiv preprint arXiv:1805.03691}, 2018.

\bibitem{duarte2012evolution}
Ana Duarte, Ido Pen, Laurent Keller, and Franz~J Weissing.
\newblock Evolution of self-organized division of labor in a response threshold
  model.
\newblock {\em Behavioral ecology and sociobiology}, 66(6):947--957, 2012.

\bibitem{erdos1952combinatorial}
Paul Erd\H{o}s and Richard Rado.
\newblock Combinatorial theorems on classifications of subsets of a given set.
\newblock {\em Proceedings of the London mathematical Society}, 3(1):417--439,
  1952.

\bibitem{georgiou2011cooperative}
Chryssis Georgiou and Alexander~A Shvartsman.
\newblock Cooperative task-oriented computing: Algorithms and complexity.
\newblock {\em Synthesis Lectures on Distributed Computing Theory},
  2(2):1--167, 2011.

\bibitem{kernbach2013adaptive}
Serge Kernbach, Dagmar H{\"a}be, Olga Kernbach, Ronald Thenius, Gerald
  Radspieler, Toshifumi Kimura, and Thomas Schmickl.
\newblock Adaptive collective decision-making in limited robot swarms without
  communication.
\newblock {\em The International Journal of Robotics Research}, 32(1):35--55,
  2013.

\bibitem{kim2014response}
Min-Hyuk Kim, Hyeoncheol Baik, and Seokcheon Lee.
\newblock Response threshold model based uav search planning and task
  allocation.
\newblock {\em Journal of Intelligent \& Robotic Systems}, 75(3-4):625--640,
  2014.

\bibitem{krieger2000ant}
Michael~JB Krieger, Jean-Bernard Billeter, and Laurent Keller.
\newblock Ant-like task allocation and recruitment in cooperative robots.
\newblock {\em Nature}, 406(6799):992, 2000.

\bibitem{lerman2006analysis}
Kristina Lerman, Chris Jones, Aram Galstyan, and Maja~J Matari{\'c}.
\newblock Analysis of dynamic task allocation in multi-robot systems.
\newblock {\em The International Journal of Robotics Research}, 25(3):225--241,
  2006.

\bibitem{macarthur2011distributed}
Kathryn~Sarah Macarthur, Ruben Stranders, Sarvapali~D Ramchurn, and Nicholas~R
  Jennings.
\newblock A distributed anytime algorithm for dynamic task allocation in
  multi-agent systems.
\newblock In {\em AAAI}, pages 701--706, 2011.

\bibitem{mclurkin2005dynamic}
James McLurkin and Daniel Yamins.
\newblock Dynamic task assignment in robot swarms.
\newblock In {\em Robotics: Science and Systems}, volume~8. Citeseer, 2005.

\bibitem{mclurkin2004stupid}
James~Dwight McLurkin.
\newblock {\em Stupid robot tricks: A behavior-based distributed algorithm
  library for programming swarms of robots}.
\newblock PhD thesis, Massachusetts Institute of Technology, 2004.

\bibitem{oster1979caste}
George~F Oster and Edward~O Wilson.
\newblock {\em Caste and ecology in the social insects}.
\newblock Princeton University Press, 1979.

\bibitem{penders2011robot}
Jacques Penders, Lyuba Alboul, Ulf Witkowski, Amir Naghsh, Joan Saez-Pons,
  Stefan Herbrechtsmeier, and Mohamed El-Habbal.
\newblock A robot swarm assisting a human fire-fighter.
\newblock {\em Advanced Robotics}, 25(1-2):93--117, 2011.

\bibitem{radeva2017costs}
Tsvetomira Radeva, Anna Dornhaus, Nancy Lynch, Radhika Nagpal, and Hsin-Hao Su.
\newblock Costs of task allocation with local feedback: Effects of colony size
  and extra workers in social insects and other multi-agent systems.
\newblock {\em PLoS computational biology}, 13(12):e1005904, 2017.

\bibitem{robinson1992regulation}
Gene~E Robinson.
\newblock Regulation of division of labor in insect societies.
\newblock {\em Annual review of entomology}, 37(1):637--665, 1992.

\bibitem{csahin2004swarm}
Erol {\c{S}}ahin.
\newblock Swarm robotics: From sources of inspiration to domains of
  application.
\newblock In {\em International workshop on swarm robotics}, pages 10--20.
  Springer, 2004.

\bibitem{su2017ant}
Hsin-Hao Su, Lili Su, Anna Dornhaus, and Nancy Lynch.
\newblock Ant-inspired dynamic task allocation via gossiping.
\newblock In {\em International Symposium on Stabilization, Safety, and
  Security of Distributed Systems}, pages 157--171. Springer, 2017.

\bibitem{wang2015multi}
Zijian Wang and Mac Schwager.
\newblock Multi-robot manipulation with no communication using only local
  measurements.
\newblock In {\em CDC}, pages 380--385, 2015.

\bibitem{yang2009swarm}
Yongming Yang, Changjiu Zhou, and Yantao Tian.
\newblock Swarm robots task allocation based on response threshold model.
\newblock In {\em 2009 4th International Conference on Autonomous Robots and
  Agents}, pages 171--176. IEEE, 2009.

\bibitem{zahugi2013oil}
Emaad Mohamed~H Zahugi, Mohamed~M Shanta, and TV~Prasad.
\newblock Oil spill cleaning up using swarm of robots.
\newblock In {\em Advances in Computing and Information Technology}, pages
  215--224. Springer, 2013.

\end{thebibliography}

\appendix

\end{document}